\documentclass[letterpaper,11pt]{article}
\def\anon{0}

\PassOptionsToPackage{table,xcdraw}{xcolor}
\usepackage[utf8]{inputenc}
\usepackage{amsmath, amsthm, amssymb, thm-restate}
\usepackage{algorithmicx}
\usepackage[ruled,vlined,linesnumbered]{algorithm2e}
\usepackage{setspace}
\usepackage{mathtools}
\usepackage[numbers]{natbib}
\usepackage{comment} 
\usepackage[most]{tcolorbox} 
\usepackage{xfrac}
\usepackage{hyperref}
\usepackage{multirow}
\usepackage{caption}
\usepackage{bm}
\usepackage{float}
\usepackage{newfloat}
\usepackage{enumitem}
\usepackage{dblfloatfix} 
\usepackage{wrapfig}
\usepackage{tikz}
\usepackage{csquotes}
\usetikzlibrary{arrows, automata, calc}
\usetikzlibrary{trees, shapes, decorations.pathreplacing, angles, quotes, decorations.pathmorphing}
\usetikzlibrary{shapes.geometric, positioning, arrows.meta}

\usepackage{fancyhdr}

\let\oldnl\nl
\newcommand{\nonl}{\renewcommand{\nl}{\let\nl\oldnl}}

\usepackage[margin=1in]{geometry}
\usetikzlibrary{fit}
\allowdisplaybreaks

\definecolor{mygreen}{RGB}{10,110,230}
\definecolor{myred}{RGB}{10,110,230}

\hypersetup{
     colorlinks=true,
     citecolor= mygreen,
     linkcolor= myred,
     urlcolor=black
}

\renewcommand{\epsilon}{\varepsilon}

\ifnum\anon=0
\newcommand{\sbcomment}[1]{{\textcolor{blue}{[\textbf{Soheil:} #1]}}}
\newcommand{\agcomment}[1]{{\textcolor{orange}{[\textbf{Alma:} #1]}}}
\newcommand{\mscomment}[1]{{\textcolor{red}{[\textbf{Madhu:} #1]}}}
\newcommand{\amcomment}[1]{{\textcolor{purple}{[\textbf{Amir:} #1]}}}
\newcommand{\vccomment}[1]{{\textcolor{green}{[\textbf{Vincent:} #1]}}}
\else
\newcommand{\sbcomment}[1]{}
\newcommand{\agcomment}[1]{}
\newcommand{\mscomment}[1]{}
\newcommand{\amcomment}[1]{}
\newcommand{\vccomment}[1]{}
\fi

\newcommand{\mA}{\mathcal{A}}
\newcommand{\mAh}{\hat{\mA}}

\newcommand{\p}{{\mathcal{P}}}
\renewcommand{\P}{{\mathcal{P}}}

\newcommand{\Z}{\mathbb{Z}}

\newcommand{\hiddencomment}[1]{}

\newcommand{\mc}[1]{\ensuremath{\mathcal{#1}}}

\newcommand{\ah}{{\hat{a}}}

\newcommand{\Ph}{{\hat{P}}}
\newcommand{\Mh}{{\hat{M}}}

\newcommand{\sink}[0]{\ensuremath{s}}
\newcommand{\QC}[0]{\ensuremath{\mathsf{QC}}}
\newcommand{\AQC}[0]{\ensuremath{\mathsf{QC}}}
\newcommand{\NAQC}[0]{\ensuremath{\mathsf{NAQC}}}
\newcommand{\tvd}{d_{\mathrm{TV}}}

\DeclarePairedDelimiter\card{\lvert}{\rvert}

\newcommand{\Yes}{{\mathrm{YES}}}
\newcommand{\No}{{\mathrm{NO}}}
\newcommand{\DYes}{{\mathcal{D}_{\Yes}}}
\newcommand{\DNo}{{\mathcal{D}_{\No}}}

\DeclareMathOperator*{\argmax}{arg\,max}

\usepackage[noabbrev,nameinlink]{cleveref}
\crefname{lemma}{Lemma}{Lemmas}
\crefname{theorem}{Theorem}{Theorems}
\crefname{property}{Property}{Properties}
\crefname{claim}{Claim}{Claims}
\crefname{definition}{Definition}{Definitions}
\crefname{observation}{Observation}{Observations}
\crefname{proposition}{Proposition}{Propositions}
\crefname{assumption}{Assumption}{Assumptions}
\crefname{line}{Line}{Lines}
\crefname{figure}{Figure}{Figures}
\creflabelformat{property}{(#1)#2#3}
\crefname{equation}{}{}
\crefname{section}{Section}{Sections}
\crefname{appendix}{Appendix}{Appendices}
\crefname{algCounter}{Algorithm}{Algorithms}
\Crefname{algCounter}{Algorithm}{Algorithms}

\newtheorem{lemma}{Lemma}[section]
\newtheorem{theorem}[lemma]{Theorem}
\newtheorem*{lemma*}{Lemma}
\newtheorem*{theorem*}{Theorem}

\newtheorem{corollary}[lemma]{Corollary}

\newtheorem{definition}[lemma]{Definition}
\newtheorem{claim}[lemma]{Claim}
\crefname{corollary}{Corollary}{Corollaries}

\newtheorem{remark}[lemma]{Remark}

\newtheorem{open}{Open Problem}
\newtheorem*{open*}{Open Problem}

\newtheorem*{problem*}{Problem}

\counterwithin{figure}{section}

\definecolor{mylightgray}{RGB}{240,240,240}

\algnewcommand{\IIf}[2]{\textbf{if} #1 \textbf{then} #2}
\algnewcommand{\EndIIf}{\unskip\ \algorithmicend\ \algorithmicif}
\algnewcommand{\IElse}[1]{\textbf{else} #1}

\newenvironment{graytbox}{
\par\addvspace{0.1cm}
\begin{tcolorbox}[width=\textwidth,
                  boxsep=5pt,
                  left=1pt,
                  right=1pt,
                  top=2pt,
                  bottom=2pt,
                  boxrule=1pt,
                  arc=0pt,
                  colback=mylightgray,
                  colframe=black,
                  ]
}{
\end{tcolorbox}
}

\newenvironment{whitetbox}{
\par\addvspace{0.1cm}
\begin{tcolorbox}[width=\textwidth,
                  boxsep=5pt,
                  left=1pt,
                  right=1pt,
                  top=2pt,
                  bottom=2pt,
                  boxrule=1pt,
                  arc=0pt,
                  colframe=black,
                  colback=white
                  ]
}{
\end{tcolorbox}
}

\newcounter{algCounter}

\let\oldlemma\lemma
\renewcommand{\lemma}{%
  \renewcommand{\emph}[1]{\textbf{##1}}
  \oldlemma}

\let\olddefinition\definition
\renewcommand{\definition}{%
  \renewcommand{\emph}[1]{\textbf{##1}}
  \olddefinition}

\ifnum\anon=0
\newcommand{\snote}[1]{{\color{brown} [Soheil: #1]}}
\newcommand{\mnote}[1]{{\color{red} [Madhu: #1]}}
\newcommand{\aanote}[1]{{\color{pink} [Amir: #1]}}
\newcommand{\agnote}[1]{{\color{blue} [Alma: #1]}}
\else  
\newcommand{\aanote}[1]{}
\newcommand{\mnote}[1]{}
\newcommand{\agnote}[1]{}
\newcommand{\snote}[1]{}
\fi 

\makeatletter
\renewcommand{\paragraph}{%
  \@startsection{paragraph}{4}%
  {\z@}{10pt}{-1em}%
  {\normalfont\normalsize\bfseries}%
}
\makeatother

\makeatletter
\patchcmd{\@algocf@start}
  {-1.5em}
  {0pt}
  {}{}
\makeatother

\title{Markov Chains with Rewinding}

\if\anon1{}\else{
\author{
Amir Azarmehr\thanks{Northeastern University. Supported in part by NSF CAREER Award CCF-2442812 and a Google
Research Award. Emails: \texttt{\{azarmehr.a, s.behnezhad, ghafari.m\}@northeastern.edu}} \and 
Soheil Behnezhad\footnotemark[1] \and
Alma Ghafari\footnotemark[1] \and
Madhu Sudan\thanks{School of Engineering and Applied Sciences, Harvard University, Cambridge, Massachusetts, USA. Supported in part by a Simons Investigator Award, NSF Award CCF 2152413 and AFOSR award FA9550-25-1-0112. Email: \texttt{madhu@cs.harvard.edu}}
}
}\fi 

\date{}

\begin{document}

\maketitle
 
\thispagestyle{empty}
 \begin{abstract}
\setlength{\parskip}{0.2cm}

Motivated by techniques developed in recent progress on lower bounds for sublinear time algorithms 
(Behnezhad, Roghani and Rubinstein, STOC 2023, FOCS 2023, and STOC 2024) we introduce and study a new class of {\em randomized algorithmic}
processes that we call ``Markov Chains with Rewinding''. In this setting an agent/{\em algorithm} interacts with a (partially observable) Markovian {\em random} evolution by periodically/strategically rewinding the Markov chain to previous states. Depending on the application this may lead the evolution to desired states faster, or allow the agent to efficiently learn or test properties of the underlying Markov chain that may be infeasible or inefficient with passive observation. 

We study the task of identifying the initial state in a given partially observable Markov chain.  Analysis of this question in specific Markov chains is the central ingredient in the above cited works and we aim to systematize the analysis in our work. Our first result is that any pair of states distinguishable with any rewinding strategy can also be distinguished with a {\em non-adaptive} rewinding strategy (i.e., one whose rewinding choices are predetermined before observing any outcomes of the chain). Therefore, while rewinding strategies can be shown to be strictly more powerful than passive strategies (i.e., those that do not rewind back to previous states), adaptivity does {\em not} give additional power to a rewinding strategy in the absence of efficiency considerations. 

The difference becomes apparent however when we introduce a natural efficiency measure, namely the {\em query complexity} (i.e., the number of observations they need to identify distinguishable states). Our second main contribution is to quantify this efficiency gap. We present a non-adaptive rewinding strategy whose query complexity is within a polynomial of that of the optimal (adaptive) strategy, and show that such a polynomial loss is necessary in general.

\end{abstract}

\newpage

{
\clearpage
\hypersetup{hidelinks}
\vspace{1cm}
\renewcommand{\baselinestretch}{0.1}
\setcounter{tocdepth}{2}
\tableofcontents{}
\thispagestyle{empty}
\clearpage
}

\setcounter{page}{1}
\section{Introduction}\label{sec:intro}

In this work, we formally define and study a new class of processes involving interaction between randomness and algorithms. In these processes, that we call \emph{Markov chains with rewinding}, a partially observable Markov chain interacts with an algorithmic agent that has the ability, at any point in time, to rewind the Markov chain to a previous point of time.
We study the problem of \emph{identifying the initial state} in this model, where the goal is for the algorithm to use the observations and the power of rewinding to identify a given (hidden) initial state.

Partially observable Markov chains form a well-established element of the toolkit in computing and arise in both the design and analysis of algorithms. Rewinding of Markov chains captures a natural ability of algorithms that work with Markov chains, and as a result, they have occurred implicitly in a variety of settings (an overview is given in \cref{sec:related-work}).

This is highlighted centrally in a class of applications of Markov chains with rewinding to {\em sublinear time algorithms}. For the maximum matching problem, \citet*{BehnezhadRR-STOC23,BehnezhadRR-FOCS23,BehnezhadRR-STOC24} prove strong lower bounds on the query complexity, and hence running time, of sublinear time algorithms by using Markov chains with rewinding to model sublinear algorithms, and then analyzing these chains. 
Recent works on non-adaptive sublinear lower bounds follow similar paradigms, while not explicitly defining a Markov chain \cite{AzarmehrBGS25,Shah26}.
Additionally, lower bounds for other graph problems in the sublinear-time model, such as edge orientation \cite{MitrovicRS24} and acyclicity \cite{BenderR00}, can also be viewed as instances of this general model. 

Despite the prevalence of this phenomenon, rewinding does not appear to have been studied formally in the literature and in particular tools to analyze the power (and limitations) of Markov chains with rewinding have been ad-hoc.
Our work is motivated principally by these and aims to build a more systematic toolkit for such analysis.

\paragraph{Markov chains with rewinding:}
We start with an example to illustrate the problem.
Consider the partially observable Markov chain depicted below, where the only observable variable is the color of the current state (i.e., whether it is in $\{a, a', b, b'\}$ or $\{s\}$).
The initial state is promised to be either $a$ or $a'$, and our goal is to identify this initial state.

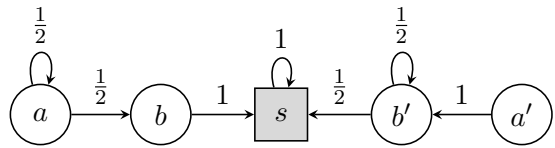
\begin{figure}[H]
    \centering
    \vspace{-0.2cm}

\begin{tikzpicture}[->, >=stealth, auto, semithick, scale=0.8]
  \tikzstyle{every state}=[fill=none,draw=black,text=black,minimum size =0.8cm]
  \tikzstyle{sink state}=[fill=gray!30,draw=black,text=black,minimum size=0.7cm]

  \node[state] (a) at (0, 0) {$a$};
  \node[state] (b) at (2, 0) {$b$};
  \node[sink state] (s) at (4, 0) {$s$};
  \node[state] (bp) at (6, 0) {$b'$};
  \node[state] (ap) at (8, 0) {$a'$};

  \path (a) edge[loop above, looseness=7] node {$\frac{1}{2}$} (a)
        (a) edge node {$\frac{1}{2}$} (b)
        (b) edge node {$1$} (s)
        (ap) edge[above] node {$1$} (bp)
        (bp) edge[loop above, looseness=7] node {$\frac{1}{2}$} (bp)
        (bp) edge[above] node {$\frac{1}{2}$} (s)
        (s) edge[loop above, looseness=7] node {$1$} (s);
\end{tikzpicture}
\end{figure}

First, note that ``passive observations'' cannot distinguish $a$ from $a'$: Whether we start from $a$ or $a'$, it takes exactly $\text{Geometric}(1/2)+1$ steps until we reach $s$ and remain there forever. But suppose we are allowed to {\em rewind} back to any previous state at any time. In this case, we can run the Markov chain for one step to reach state $X$, and then run $D$ independent steps from $X$ for a sufficiently large $D$. The idea is that if $X = b'$ (which implies the initial state is $a'$), then it will most likely have a mix of $s$ and not $s$ next steps, whereas if $X \in \{a, b\}$ (implying that the initial state is $a$) all next steps will be the same. The following figure illustrates this.

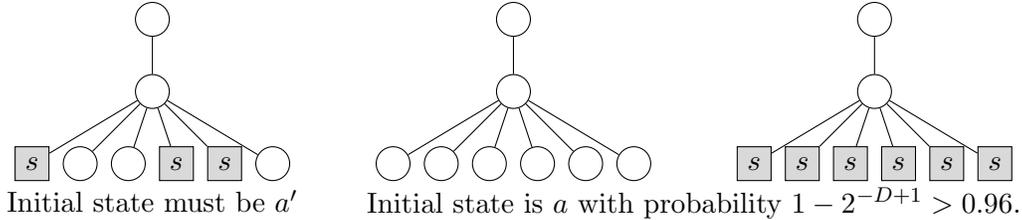
\begin{figure}[H]
    \centering
\begin{tikzpicture}[
  scale=0.8,
  level distance=12mm,
  level 1/.style={sibling distance=32mm},
  level 2/.style={sibling distance=8mm},
  circ/.style={draw, circle, minimum size=4.5mm, inner sep=0pt, fill=white},
  sq/.style={draw, rectangle, minimum size=4.5mm, inner sep=0pt, fill=gray!30}
]

\begin{scope}[shift={(0,0)}]
\node[circ] (r1) {}
  child { node[circ] {}
    child { node[sq]{$s$} }
    child { node[circ]{} }
    child { node[circ]{} }
    child { node[sq]{$s$} }
    child { node[sq]{$s$} }
    child { node[circ]{} }
  };
\node[below=19mm of r1] {Initial state must be $a'$};
\end{scope}

\begin{scope}[shift={(6,0)}]
\node[circ] (r2) {}
  child { node[circ] {}
    child { node[circ]{} }
    child { node[circ]{} }
    child { node[circ]{} }
    child { node[circ]{} }
    child { node[circ]{} }
    child { node[circ]{} }
  };
\end{scope}

\begin{scope}[shift={(12,0)}]
\node[circ] (r3) {}
  child { node[circ] {}
    child { node[sq]{$s$} }
    child { node[sq]{$s$} }
    child { node[sq]{$s$} }
    child { node[sq]{$s$} }
    child { node[sq]{$s$} }
    child { node[sq]{$s$} }
  };
\end{scope}

\node at ($(r2)!0.5!(r3) + (0,-30.5mm)$) {Initial state is $a$ with probability $1-2^{-D+1}> 0.96.$};

\end{tikzpicture}
\end{figure}

We now present a more general description of the model, deferring its formalization to \cref{sec:model}.
The input consists of a partially observable Markov chain $M$, two candidate initial states $a$ and $a'$, and a hidden initial state $X_0 \in \{a, a'\}$.
The algorithm interacts with the Markov chain through the hidden initial state $X_0$.
It can either draw states according to the transition probabilities, or rewind the state to an earlier point in time.
The drawn states $\{X_t\}_t$ are hidden from the algorithm, but observations $\{Z_t\}_t$ are available, where $Z_t = O(X_t)$.
Eventually, the algorithm must decide whether the initial state $X_0$ is $a$ or $a'$, based on the observations.
We study the runtime and the number of queries (i.e.\ drawn states) required to do so.
Other than the motivation from sublinear-algorithm lower bounds, rewinding formalizes a natural approach to interacting with Markov chains; variations and special cases have been considered repeatedly in previous work (see \cref{sec:related-work}).

\subsection{Our Contributions}

Our main focus in this work is on the following fundamental problem: 
\begin{quote}{\em Given a partially observable Markov chain, can a nearly-optimal rewinding strategy be computed efficiently?}\end{quote}

To better understand this question, we first need to formalize the model. We present the formalization in \cref{sec:model}; We define \emph{rewinding strategies} 
and \emph{state identification algorithms}.
Rewinding strategy refers to the choice of rewinding times throughout the process.
A state identification algorithm computes the rewinding strategy through which it interacts with the partially observable Markov chain, and ultimately identifies the initial state based on the observations.
We study and compare the power of {\em adaptive} and {\em non-adaptive} rewinding strategies for this task, where a non-adaptive strategy determines how to rewind independently of the observations. 
These lead us to our key complexity measures $\AQC_M(a,b)$ (resp.\ $\NAQC_M(a,b)$) for the {adaptive} (resp.\ non-adaptive) {\em query complexity} of distinguishing state $a$ from $b$ in chain $M$,
i.e.\ the minimum number of observations required by a rewinding strategy to distinguish the two states.
Query complexity provides a benchmark for measuring the efficiency of an algorithm, and a means of proving lower bounds for its runtime. 

We show that there is indeed a polynomial-time algorithm for the question stated at the beginning of this section. Formally, we prove the following theorem.

\begin{graytbox}
    \begin{restatable}{theorem}{thmnaub}
    \label{thm:NA-ub}
    Given a partially observable Markov chain $M = (\Omega, P, O)$, there exists a non-adaptive algorithm and a constant $c = c(\card{\Omega})$ that can distinguish between any two states $a, b \in \Omega$ in time 
        $$
        c \cdot \QC_M(a, b)^{O(|\Omega|)}.
        $$
    \end{restatable}
\end{graytbox}

In the statement above, we let $M$ be a Markov chain on state space $\Omega$ with a special {\em sink} state. The only observable signal available to the rewinding algorithm is whether the chain is in the sink state or not (similar to the example in \cref{sec:intro}). We call such Markov chains {\em canonical}
and show that identifying the initial state in canonical Markov chains is essentially as hard as general partially observable Markov chains (see \cref{sec:generality}). 

Observe that if non-adaptive algorithms cannot distinguish two states $a$ and $b$ of a Markov chain $M$, then \cref{thm:NA-ub} implies that no adaptive algorithm can do so either. Therefore, non-adaptive algorithms are as powerful as adaptive ones in terms of which pairs of states they can distinguish. 

Furthermore, formalized as \cref{thm:NA-A-gap}, we show that the (polynomial) power of $\QC_M(a, b)$ in \cref{thm:NA-ub} is tight, up to a constant, for the runtime of non-adaptive algorithms. More specifically, we present a partially observable Markov chain $M$ and states $a$ and $b$, where a polynomial gap exists between the adaptive and non-adaptive query complexity of distinguishing $a$ and $b$. 

\begin{graytbox}
    \begin{theorem}\label{thm:NA-A-gap}
        There is a choice of $M$ and pair $a, b \in \Omega$ such that for some $c = c(\card{\Omega})$,
        $$
         \NAQC_M(a, b) \geq c \cdot \AQC_M(a, b)^{(1-o(1))|\Omega|}.
        $$
        This holds for infinitely many choices of $|\Omega|$ and query complexities.
    \end{theorem}
\end{graytbox}

We remark that the adaptive rewinding strategy presented for this gap can be easily turned into an algorithm with the same run time.
Therefore, the polynomial gap is also present between the optimal adaptive and non-adaptive run times of state identification algorithms.

We leave open the question of whether the number of states in the Markov chain can be removed from the run-time exponent. Any solution would require new techniques and necessarily rely on adaptivity, as remarked above.
Formally, we ask:

\begin{open}
    Given a partially observable Markov chain $M = (\Omega, P, O)$, candidate initial states $a, b \in \Omega$, and hidden initial state $X_0 \in \{a, b\}$, is there an adaptive state identification algorithm that successfully identifies the initial state in time $\AQC_M(a, b)^{C}$ where constant $C$ is independent of the number of states?
\end{open}

We believe it is also interesting to compare non-adaptive algorithms against the best non-adaptive benchmark for future work. In particular, we ask: 

\begin{open}
    Given a partially observable Markov chain $M = (\Omega, P, O)$, candidate initial states $a, b \in \Omega$, and hidden initial state $X_0 \in \{a, b\}$, is there a non-adaptive state identification algorithm that successfully identifies the initial state in time $\NAQC_M(a, b)^C$ where constant $C$ is independent of the number of states?
\end{open}

\subsection{Related Work} 
\label{sec:related-work}

We overview settings where partially observable Markov chains with rewinding have occurred naturally.
A highly prevalent instance is the folklore method of converting expected {\em run times of algorithms} into high probability bounds by restarting the algorithm when it does not halt promptly.
This is a special case of rewinding where the underlying state of the algorithm is captured by a Markov chain, and the observation is limited to knowing whether the algorithm has halted or not. 

Restarting/rewinding can also help boost the success of algorithms as exemplified in the famed randomized algorithm for 3SAT due to Sch\"oning~\cite{Schoning99}.
In the case of random walks on multi-dimensional lattices, \citet*{JansonP12} showed that the expected hitting time of the origin is finite when the algorithm is allowed to restart the walk back from the starting point.
More general forms of rewinding lead to elegant new approaches to {\em optimization} and {\em random sampling} as in the ``Go-with-the-winners'' algorithm of Aldous and Vazirani~\cite{AldousVazirani}. 
Selecting which state to continue the process from is also studied by \citet*{DumitriuTW03}, where the algorithm decides which Markov chain to advance, in order to minimize the expected time that one of them reaches the target state.
In {\em cryptography}, rewinding arguments play a celebrated role in proving the security of protocols starting with the seminal work of Goldwasser, Micali, and Rackoff~\cite{GoldwasserMR89}. Rewinding is also a central element in defining strong notions of security as in the notion of ``resettable zero knowledge'' due to Canetti, Goldreich, Goldwasser, and Micali~\cite{CanettiGGM00}.

\subsection{Connection to Sublinear-Time Graph Algorithms}
\label{sec:connections-to-sublinear}

To further motivate our model, we overview its connection to sublinear time graph algorithms (or more precisely, lower bounds).
Here, we provide a high-level description of this connection in general. 
Later in \cref{sec:sublinear}, we present an explicit connection to the lower bound for testing acyclicity. That is, we show how the lower bound is captured by Markov chains with rewinding.

Let us first briefly overview the model of sublinear-time algorithms for graphs. We are given a graph $G=(V, E)$ to which we have adjacency list query access.\footnote{Another common access model is the adjacency matrix model.} Each query can specify a vertex $v$ and an integer $i \geq 1$. The answer to such a query is the $i$-th neighbor of vertex $v$ stored in an arbitrarily ordered list, or $\perp$ if $\deg(v) < i$. Now, the main question of interest, is the number of such queries needed to estimate a property of the graph. For example, the works of \cite{BehnezhadRR-STOC23,BehnezhadRR-FOCS23,BehnezhadRR-STOC24} focus on proving lower bounds on the query complexity of estimating the size of maximum matching in the graph (see also \cite{Behnezhad21,YoshidaYI09,NguyenO08,BhattacharyaKS23} for some algorithms). Many other problems can, and have been, studied in this setting, including graph coloring \cite{AssadiCK19}, correlation clustering \cite{Assadi022}, metric properties such as minimum spanning trees \cite{CzumajS04,ChazelleRT01}, Steiner trees \cite{CKT23}, TSP \cite{CKT-ICALP23,BehnezhadRRS-ICALP24}, and many others.

Remarkably, for many graph properties, we have had {\em unconditional} query lower bounds that can be matched algorithmically with sublinear time algorithms. However, these lower bound arguments are often ad hoc, and there is a lack of systematic tools to prove them. Our hope is that our model and techniques can serve as a first step towards developing such generic lower bound tools. Let us expand on this connection a bit further.

 Typically, sublinear time lower bounds are proved by specifying two distributions $\mc{D}_{YES}$ and $\mc{D}_{NO}$ of graphs, where graphs drawn from $\mc{D}_{YES}$ have the desired property and those drawn from $\mc{D}_{NO}$ do not. The lower bound then follows by showing that distinguishing these distributions requires many queries to the graph. In most cases, these distributions are constructed by having a few groups $A_1, A_2, \ldots, A_k$ that each vertex may belong to and is kept hidden from the algorithm. The distribution of edges between these groups differs in $\mc{D}_{YES}$ and $\mc{D}_{NO}$, and thus the task of distinguishing these distributions reduces to identifying these groups. The only useful signal that the algorithm receives, and helps in identifying the groups, is the vertex degrees. This is precisely how the hard instances of \cite{BehnezhadRR-STOC23,BehnezhadRR-FOCS23,BehnezhadRR-STOC24} for maximum matchings are constructed. Lower bounds for other graph problems in the sublinear-time model, such as edge orientation \cite{MitrovicRS24} and acyclicity \cite{BenderR00}, can also be viewed as instances of this general framework.

 Now note that each group $A_i$ can be viewed as a different state of the Markov chain. We let $p_{ij}$ (i.e., the transition probability from $A_i$ to $A_j$) to be the fraction of edges of group $A_i$ that go to $A_j$, representing the fact that a random adjacency list neighbor of an $A_i$ node belongs to $A_j$ with probability $p_{ij}$. We construct these states separately for groups in $\mc{D}_{YES}$ and $\mc{D}_{NO}$. Then, assuming that the algorithm cannot find cycles (which holds in many lower bounds such as \cite{BehnezhadRR-STOC23,BehnezhadRR-FOCS23,GoldreichR97,BenderR00}) the task of distinguishing the two distributions reduces to distinguishing if the initial state is among the YES states or the NO states of the Markov chain. Note, in particular, that because the sublinear time algorithm is not constrained to follow a single path in the graph and can adaptively make adjacency list queries to any previously visited vertex, this corresponds to a rewinding strategy on the Markov chain. Therefore, the study of the power and limitations of such rewinding strategies for general Markov chains closely relates to the query complexity of sublinear-time graph algorithms, and can serve as a rich toolkit for proving lower bounds against them.

\subsection*{Paper Arrangement}
More examples of Markov chains with rewinding are presented in \cref{sec:examples}.
An overview of our techniques appears in \cref{sec:techniques}.
The model is formally defined in \cref{sec:model}.
We examine state partitions, a key tool for our algorithm, in \cref{sec:partitions}.
\cref{thm:NA-ub,thm:NA-A-gap} are proven in \cref{sec:prim,sec:gap} respectively.
Finally, in \cref{sec:generality}, we prove the generality of canonical Markov chains.

\section{Motivating Examples \& Technical Overview}
\label{sec:examples}

We start this section by presenting examples of rewindable Markov chains along with optimal strategies to distinguish between two candidate initial states. 
These examples further clarify the model and demonstrate how rewinding strategies can be \emph{non-trivially} effective, and thus hard to prove lower bounds against.
After these examples, we outline the techniques used in our algorithms and analyses.

\subsection{Example 1: The non-trivial power of (non-adaptive) rewinding}
\label{sec:example-stoc23}
\label{sec:example-1}

\begin{wrapfigure}[9]{r}{0.35\textwidth}  
    \centering
    
  \vspace{-0.9cm}
\begin{center}
\resizebox{!}{4cm}{
\begin{tikzpicture}[->, >=stealth', auto, semithick, scale=0.9]
  \tikzstyle{every state}=[fill=none,draw=black,text=black]
  \tikzstyle{sink state}=[fill=gray!30,draw=black,text=black,minimum size=0.8cm]

  \node[state] (a) at (0, 0) {$a$};
  \node[state] (b) at (3, 0) {$b$};
  \node[state] (ap) at (0, -3) {$a'$};
  \node[state] (bp) at (3, -3) {$b'$};
  \node[sink state] (sink) at (5.5, -1.5) {$s$};

  \path (a) edge[bend left] node {1} (b)
        (b) edge[bend left] node {$1 - \frac{1}{d}$} (a)
        (b) edge node[below] {$\frac{1}{d}$} (sink)
        (ap) edge[loop left] node {$\frac{1}{d}$} (ap)
        (ap) edge[bend left] node[above] {$1 - \frac{1}{d}$} (bp)
        (bp) edge[bend left] node[below] {$1 - \frac{1}{d}$} (ap)
        (bp) edge node[above] {$\frac{1}{d}$} (sink)
        (sink) edge[loop right] node {$1$} (sink);
\end{tikzpicture}
}
\end{center}
\end{wrapfigure}

Our first example involves a simple Markov chain on 5 states, inspired by the hard distribution of \cite{BehnezhadRR-STOC23} for graph matchings. 
The chain is parameterized by an integer $d$, and its transition probabilities are multiples of $1/d$ as illustrated by the arrows in the figure. 
There is a special sink state $\sink$, and the algorithm can only observe whether a state is sink or non-sink.
The goal is to distinguish whether the initial state is $a$ or $a'$.
We show first how the natural strategy to distinguish two candidate initial states $a$ and $a'$ leads to query complexity $\widetilde{O}(d^2)$. A quick look at the chain and the rewinding strategy seems to suggest that such a $d^2$ complexity is essential; but we show later that a more sophisticated strategy actually distinguishes $a$ from $a'$ in $O(d)$ queries.

\paragraph{A simple algorithm with $\widetilde{O}(d^2)$ query complexity:} We say a node in the tree is an $A$-node if its state is $a$ or $a'$ and a $B$-node if its state is $b$ or $b'$.  First, we claim that given a node $v$ in the tree, we can determine with probability $1-\epsilon$ if it is an $A$ or $B$ node with $O(d \log(1/\epsilon))$ queries---we call this an $A$/$B$ test. 
To do this, note that if $v$ is an $A$-node, any child of $v$ is non-sink (specifically, it is a $B$-node). Moreover, if $v$ is a $B$-node, each child has a probability of $1/d$ of being a sink node. Therefore, opening $O(d\log(1/\epsilon))$ children for a node $v$ is sufficient for $A$/$B$ testing $v$.

 \begin{wrapfigure}[7]{r}{0.35\textwidth}  
    \centering
    
  \vspace{-0.7cm}
  
  \resizebox{!}{3cm}{
     \begin{tikzpicture}[
      level 1/.style={sibling distance=20mm},
      level 2/.style={sibling distance=10mm},
      edge from parent/.style={draw, -latex},
      every node/.style={circle, draw, minimum size=4mm},
      no border/.style={draw=none},
      scale=0.8
    ]
    
    \node {} 
      child {node {}
        child {node {}}
        child {node[draw=none] {...}}
        child {node {}}
      }
      child {node[draw=none] {...}}
      child {node {}
        child {node {}}
        child {node[draw=none] {...}}
        child {node {}}
      };
    
    \draw [-] (-0.7,0) arc[start angle=180,end angle=360,radius=0.7cm];
    
    \draw [-] (1.3,-1.5) arc[start angle=180,end angle=360,radius=0.7cm];
    
    \node[draw=none] at (1.8, 0) {$\Theta(d \log d)$};
    \node[draw=none] at (3.8, -1.5) {$\Theta(d \log d)$};
    
    \end{tikzpicture}
    }

  \vspace{-0.2cm}
\end{wrapfigure}

 Next, having the $A$/$B$ test available to us, we aim to assert whether a given $A$-node $v$ has state $a$ or $a'$. To do so, we first open $d \log (1/\epsilon)$ children for $v$. If $v$ has state $a$, then all of its children must be $B$-nodes. On the other hand, if $v$ has state $a'$, each child is an $A$-node with probability $1/d$, and so at least one must be an $A$-node with probability $1-(1-1/d)^{d\log(1/\epsilon)} \geq 1-\epsilon$. It then suffices to run the $A$/$B$ test on these children to check if any of them is an $A$ node. This takes total time $O(d^2 \log^2(1/\epsilon))$. Choosing $\epsilon = 1/d^3$  so that all $A$/$B$ tests are correct w.h.p., we get an overall query complexity of $O(d^2 \log^2 d)$. The final queried tree is non-adaptive and is illustrated on the right.

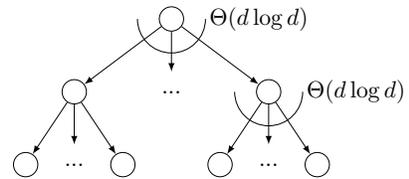
\begin{wrapfigure}[9]{r}{0.35\textwidth}  
    \centering
    
  \vspace{-0.3cm}
    \resizebox{!}{4.5cm}{ 
    \begin{tikzpicture}[
      sibling distance=20mm,
      edge from parent/.style={draw, -latex},
      every node/.style={circle, draw, minimum size=4mm},
      no border/.style={draw=none},
      scale=0.7
    ]
    
    \node {} 
      child {node {}
                      child {node[draw=none] {...}
                        child {node {}
                          child {node {}}
                          child {node[draw=none] {...}}
                          child {node {}}
                        }
        }
      };
    
    \draw [-] (-0.7,-4.5) arc[start angle=180,end angle=360,radius=0.7cm];
    \node[draw=none] at (4.5, -4.5) {$\Theta(d)$ (for $A$/$B$ testing)};
    
    \draw [decorate,decoration={brace,mirror,amplitude=10pt},xshift=-4pt] 
      (-0.3,-0.2) -- (-0.3,-4.5) node [midway,xshift=-15pt,draw=none] {};
    \node[draw=none] at (-1.4, -2.3) {$2d$};
    
    \end{tikzpicture}
    }
  \end{wrapfigure}
  
 \paragraph{Improving query complexity to $O(d)$:} We now describe how we can improve over the algorithm above, and distinguish the $a$ and $a'$ states with just $O(d)$ queries. Here is the key insight. Suppose that we open a path of length $2d$ from the root and let us condition on not seeing the $\sink$ node at all (which happens with constant probability). Then, if we start from the $a$ state, we must alternatively visit $a$ and $b$, finally arriving back at $a$ since the length of the path is even. However, if we start from $a'$, there is a constant probability that we go from $a'$ to $a'$ exactly once throughout the process since $P(a', a')=1/d$ and the length of the walk is $2d$. If this event happens, we end up in state $b'$. Therefore, it suffices to run the $A$/$B$ test on the last vertex of the path---if it happens to be $B$, we will be sure that we started from $a'$. Note that the success probability is a small constant, but we can boost it by repeating this process multiple times. The figure on the right shows the final query tree (for a single repetition of this process), which again is non-adaptive.

 \subsection{Example 2: The power of adaptivity}
\label{sec:example-2}
Our second example shows the power of adaptive rewinding strategies compared to non-adaptive ones. This example, illustrated in \cref{fig:Gap-A-NA}, is later used to prove a polynomial gap between the non-adaptive and adaptive query complexities in the worst case (\cref{thm:NA-A-gap}). We formalize this in \cref{sec:gap}, but here provide an intuition of how adaptivity helps. 

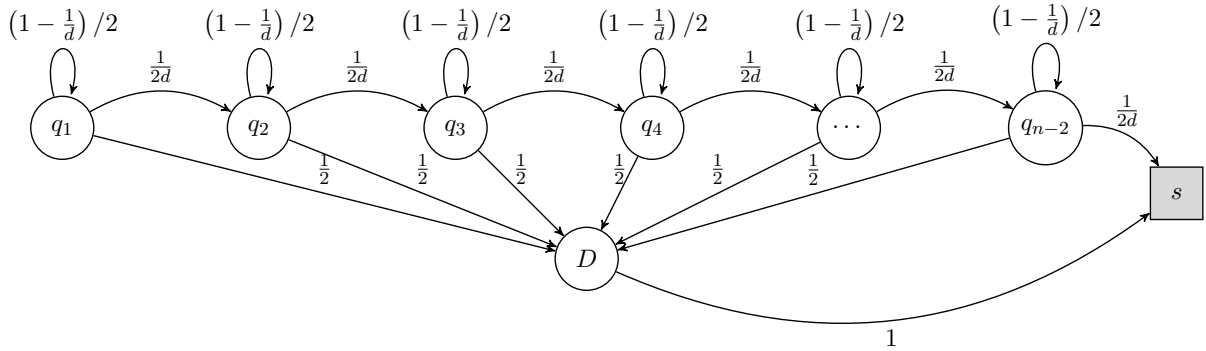
\begin{figure}
    \centering
    \resizebox{!}{4.8cm}{\begin{tikzpicture}[->, >=stealth', auto, semithick, scale=2]
  \tikzstyle{every state}=[fill=none,draw=black,text=black]
  \tikzstyle{sink state} = [
            rectangle,
            fill=gray!30,
            draw=black,
            minimum size=0.8cm,
            text centered
        ]
  \node[state] (q_1) at (0, 0)  {$q_1$};
  \node[state] (q_2) at (1.5, 0) {$q_2$};
  \node[state] (q_3) at (3, 0) {$q_3$};
  \node[state] (q_4) at (4.5, 0) {$q_4$};
  \node[state] (dots) at (6, 0) {$\cdots$};
  \node[state] (q_n) at (7.5, 0) {$q_{n-2}$};
  \node[sink state] (sink) at (8.5, -0.5) {$s$};
  \node[state] (D) at (4, -1) {$D$};

  \path[->]
  
    (q_1) edge[bend left] node[midway,above] {$\frac{1}{2d}$} (q_2)

    edge[loop above] node {$\left(1-\frac{1}{d}\right)/2$} (q_1)
          
    (q_2) edge[bend left] node[midway,above] {$\frac{1}{2d}$}  (q_3)
          edge[loop above] node {$\left(1-\frac{1}{d}\right)/2$} (q_2)
    (q_3) edge[bend left] node[midway,above] {$\frac{1}{2d}$}  (q_4)
          edge[loop above] node {$\left(1-\frac{1}{d}\right)/2$} (q_3)

    (q_4) edge[bend left] node[midway,above] {$\frac{1}{2d}$}  (dots)
          edge[loop above] node {$\left(1-\frac{1}{d}\right)/2$} (q_4)
    
    (q_n) edge[bend left] node[midway,above] {$\frac{1}{2d}$}  (sink)
          edge[loop above] node {$\left(1-\frac{1}{d}\right)/2$} (q_n)
    (D) edge[bend right] node[midway,below] {$1$} (sink)

    (dots) edge[bend left] node[midway,above] {$\frac{1}{2d}$} (q_n)
          edge[loop above] node {$\left(1-\frac{1}{d}\right)/2$}(dots)

    (q_1) edge[right] node[midway,above] {$\frac{1}{2}$} (D)
            
    (q_2) edge[right] node[midway,above] {$\frac{1}{2}$} (D)
            
    (q_3) edge[right] node[midway,above] {$\frac{1}{2}$} (D)
    (q_4) edge[right] node[midway,above] {$\frac{1}{2}$} (D)

    (q_n) edge[right] node[midway,above] {$\frac{1}{2}$} (D)

    (dots) edge[right] node[midway,above] {$\frac{1}{2}$} (D);

\end{tikzpicture}}
    \caption{A canonical Markov chain for which there is a (large) polynomial gap between adaptive and non-adaptive strategies of distinguishing states $q_1$ and $q_2$.}
    \label{fig:Gap-A-NA}
\end{figure}

Given a parameter $d \geq 1$, the Markov chain is formed by combining a \enquote{simple Markov chain} and a \enquote{dummy state}.
The simple Markov chain consists of a path of states, where each state transitions to the next with probability $\frac{1}{d}$, and stays in the same state otherwise.
States can be easily distinguished in the simple Markov chain using $O(n^2 d)$ queries, where $n = |\Omega|$ is the number of states.
To do so, it suffices to take random walks from the initial state and examine the time it takes to move to the sink (see \cref{sec:gap} for the proof).

However, when combined with the dummy state, every state has a $\frac{1}{2}$ chance of jumping to the dummy state (otherwise, it proceeds as before).
Then, the dummy state directly moves to the sink with probability $1$.
As all the states have the same probability of transitioning to the dummy state, moving to the sink through the dummy state reveals no information about the initial state.

Adaptive strategies can easily filter out the dummy state. 
That is, they can simulate a random walk on the simple Markov chain by checking if the walk has transitioned to the dummy state along the way, and rewinding back to the previous step if it has.
As a result, adaptive strategies retain the $O(n^2 d)$ query complexity even with the dummy state added.

The non-adaptive strategies cannot do the same,
i.e.\ they cannot check whether a random walk has transitioned to the dummy state on the fly (they can only infer it after all the queries are made).
Intuitively, to make up for this, they have to take many more random walks from the initial state, so as to guarantee that plenty of them reach the sink without moving through the dummy state. Therefore, the non-adaptive query complexity suffers when the dummy state is added and becomes (approximately) $\Omega(d^{n})$. See \cref{sec:gap} for the proof.

\subsection{Technical Overview} \label{sec:techniques}

Recall that our main goal is to bring a systematic understanding of distinguishing two candidate initial states in a known Markov chain. 
In this section, we provide a brief overview of our methods. 
There are two main aspects associated with our task: (1) computing rewinding strategies to distinguish between two initial states, and (2) proving lower bounds for the query complexity of this task. Together the two aspects allow us to establish the efficiency of our algorithm, i.e.\ that its run time is competitive with the optimal query complexity.
Our approach to both aspects heavily utilizes partitions of the state set $\Omega$.

Let $\P$ be a partition of the states,
and for a state $a \in \Omega$ let $\P(a)$ denote the ``class'' of $a$ in $\P$, i.e., the set within the partition $\P$ that contains $a$. 
As a subroutine, we study the problem of identifying the class of any given (hidden) state in $\P$, i.e., the task of computing $\P(X_0)$ for a given chain $M$ and partition $\P$ (note that this is a generalization of the A/B test in \cref{sec:example-stoc23}). 
By choosing $\P$ carefully, this allows us to distinguish between two states $a, b \in \Omega$. Specifically,  it suffices to identify the class of the initial state for some partition $\P$ that separates $a$ and $b$, i.e., where $\P(a) \ne \P(b)$.

Intuitively, identifying the class within a given partition becomes more difficult as the partition becomes more refined.
As an extreme case, the test for the trivial partition $\P_0 = \{ \Omega \setminus \{\sink\}, \{\sink\}\}$ can be performed 
without any further queries
in a canonical Markov chain since the partial observation at any state determines its class in $\P_0$.  
Our key tool in the design of algorithms is to assume that for some partition $\P_1$ we know how to identify the states, and then to use this algorithm to build a more refined partition $\P_2$ where we can identify states. Starting this with the trivial partition, we can 
use this method to work our way up to more and more refined partitions, till we reach a partition that separates $a$ and $b$ at which stage our problem would be solved. 
The exact details of this refinement step (and how we estimate their query complexity)
are described below.

Recall that the goal of  \cref{thm:NA-ub} is to 
develop a non-adaptive algorithm that distinguishes between two states $a$ and $b$ in time $O_n(1) \cdot \AQC_M(a, b)^{O(n)}$, where $n$ is the number of states and $O_n(1)$ suppresses terms dependent only on $n$.
To do so, we define a weighted directed graph on the partitions: Namely, every partition $\P$ of $\Omega$ is a vertex of this graph and there is a weighted edge between
every pair of partitions $\P_1 \to \P_2$ where $\P_2$ refines $\P_1$. 
The weights $w(\P_1,\P_2)$ are formally defined shortly.
Intuitively, the main idea is to get a quantity that roughly allows us to determine the class of the initial state in $\P_2$ using approximately $w(\P_1,\P_2)$ calls to the class-identifier for $\P_1$. 
Note that if there is a path $\P_1\to\P_2 \to \P_3$ in this graph, then this amounts to saying that classes of $\P_3$ can be identified using roughly $w(\P_1,\P_2)\cdot w(\P_2,\P_3)$ calls to the class identifier for $\P_1$. Thus, given this graph, we can easily devise a non-adaptive strategy for distinguishing between $a$ and $b$. 
 
 For every partition $\P$ we compute the minimum distance from the trivial partition $\P_0$ (using this multiplicative measure of length of a path), and then use the partition $\P$ that separates $a$ from $b$ and has the smallest distance among all such $\P$. 
In particular, if we let $\P_0, \P_1, \ldots, \P_k = \P$ be the shortest path to $\P$, the test for each $\P_i$ can be recursively performed by calling the tests $\P_{i-1}$.
This results in a test for $\P_k$ (which recall differentiates between $a$ and $b$) that uses $c(\P_k) = \prod_{1 \leq i \leq k} w(\P_{i-1}, \P_i)$ tests.
For \cref{thm:NA-ub}, we also need to prove $c(\P_k) \leq\AQC_M(a, b)^{O(n)}$ which we discuss shortly, but first we describe how we compute $w(\P_1,\P_2)$. 

To understand how a test for $\P_1$ can be used to identify the class in $\P_2$ of a given state,
consider, as a warm-up, the original problem of distinguishing between $a$ and $b$.
Equipped with the test for $\P_1$, one can differentiate between $a$ and $b$ as follows:
Given a state $x \in \{a, b\}$, we take many samples from the next state (rewinding back to $x$ after each sample) and examine the distribution of their classes in $\P_1$.
If the distribution of these classes for $a$ and $b$ are $\delta$-apart in total variation distance,
then $x$ can be identified successfully with probability $1 - \epsilon$ using $O(\log(1/\epsilon)/\delta^2)$ samples.
A similar approach can be taken to distinguish classes of $\P_2$ (instead of two states $a$ and $b$).
This is again captured by the total variation distance of certain distributions which gives us the weight $w(\P_1, \P_2)$ in the graph. Specifically, we can take 
$$w(\P_1,\P_2) = \max_{\substack{c, d: \\ \P_1(c) = \P_1(d) \\ \P_2(c) \neq \P_2(d)}} \frac{1}{\tvd^\P(c, d)^2},$$ 
where for $x \in \{c,d\}$, and $\tvd^\P(c, d)$ denotes the total variation distance between next states from $c$ and $d$ when projected on $\P$. 
That is, letting $X^c$ and $X^d$ be the next state of the Markov chain conditioned on the last state being $c$ and $d$ respectively, $\tvd^P(c, d) := \tvd(\P(X^c), \P(X^d))$. 
Thus, this expression captures the hardest to separate pair of states $c$ and $d$ from different classes of $\P_2$ (and not already separated in $\P_1$) after taking one step on the Markov chain starting at these states, using only a class identifier for $\P_1$.

Turning to our lower bound, recall that we wish to express the shortest path length in terms of the query complexity, i.e.\ to prove that $\prod_{1 \leq i \leq k} w(\P_{i-1}, \P_i)$ is at most  ${O_n(1) \cdot (\AQC_M(a,b))^{O(\card{\Omega})}} $, where $\P_0\to\P_1 \cdots \to \P_k$ is the shortest path found by our algorithm. 
The argument has two main components, given a partition $\P$:
(1) if two states $c$ and $d$ with $\P(c) = \P(d)$ are far in terms of $\tvd^\P$, then $\P$ has a small outgoing edge separating $c$ and $d$, and
(2) if every two states within the same class of $\P$ are close in terms of $\tvd^\P$, then it is difficult to distinguish any two states in the same class of $\P$.
Note, in particular, that if such a $\P$ does not separate $a$ and $b$, then it would serve as a \enquote{certificate} for the difficulty of distinguishing between them.

We begin by explaining how these components can be applied to obtain upper bounds on the shortest path (equivalently, lower bounds on the query complexity).
Consider any partition $\P$, where $a$ and $b$ are not separated.
Given that $a$ and $b$ are distinguishable with $\QC_M(a, b)$ queries, we can employ (2) to show there are two states $c$ and $d$ with $\P(c) = \P(d)$ such that $\tvd^\P(c, d) = \Omega(1/\QC_M(a,b))$. 
Then, using (1), we can infer that $\P$ has an outgoing edge of weight at most $O_n(1) / \tvd^\P(c, d)^2 = O_n(1) \cdot \QC_M(a, b)^2$.
As a result, we can start at $\P_0$ and iteratively take the smallest outgoing edge, which has weight at most $O_n(1) \cdot \QC_M(a, b)^2$, until we reach a partition that separates $a$ and $b$.
The path can take a length of at most $n-1$ since each partition is refined by the next.
This gives us a path of length at most $O_n(1) \cdot \QC_M(a, b)^{2n}$ from $\P_0$ to a partition that separates $a$ and $b$.

Now we discuss each component separately.
First, consider a partition $\P$ with that two states $c$ and $d$ in the same class of $\P$ that have $\tvd^\P(c, d) = D$.
Let $G$ be an auxiliary graph where the vertices are the Markov chain states, and two states are connected if they are closer than $D / n$ in $\tvd^\P$.
We construct a refinement $\P'$ of $\P$ by letting two states be in the same class of $\P'$ if they are in the same class of $\P$ and the same connected component of $G$.
Observe that $\P'$ separates $c$ and $d$ since $\tvd^\P$ forms a metric, and $c$ and $d$ are more are $D$ apart, whereas any two directly connected states are at most $D / n$ apart.
Furthermore, the only states that are newly separated by $\P$ are at least $D/n$ apart.
Therefore, $\P$ has an outgoing edge of weight at most $(n / D)^2$ to $\P'$.

For the second component, take a partition $\P$ where any two states in the same class of $\P$ are at most $\theta$ apart in $\tvd^\P$. We show that for any two states $a$ and $b$ with $\P(a) = \P(b)$, it holds $\QC_M(a, b) = \Omega(1/\theta)$.
Given two possible initial states $a$ and $b$ to distinguish,
we consider two instances of running the rewinding strategy: one with $a$, and one with $b$.
Let $T_a$ and $T_b$ be the query trees that the rewinding strategy creates in each case.
We couple $T_a$ and $T_b$ such that they look the same to the strategy with a large probability.
Here, by looking the same, we mean that the states generated in $T_a$ and $T_b$ yield the same observations, i.e.\ they are both sink or non-sink.
In addition, the coupling is such that $T_a$ and $T_b$ are isomorphic w.r.t.\ $\P$, meaning the corresponding nodes in $T_a$ and $T_b$ have states that are in the same class of $\P$, with a large probability.

We construct our couplings in an inductive manner.
In every step, a state $u_a \in T_a$ and the corresponding state $u_b \in T_b$ are chosen by the rewinding strategy to derive the next step.
At this point, we need a method to couple the next step drawn from $u_a$ and $u_b$ such that they are both sink or non-sink with a large probability.
Such a coupling is not possible for arbitrary states $u_a$ and $u_b$ (e.g.\ $\P_0(u_a) = \P_0(u_b)$ does not imply the same is true for their next step with a large probability).
To do so, we employ the additional property that $\P(a) = \P(b)$, and show that the drawn states can be coupled such that they are in the same class of $\P$ with probability $1 - \theta$.
Intuitively, if a rewinding strategy distinguishes between $a$ and $b$, it should make at least $\Omega(1 / \theta)$, so that the above scheme \enquote{fails} and the drawn states are in different classes of $\P$.

We remark that something as strong as the shortest path algorithm is not required for \cref{thm:NA-ub}. While it is intuitive to use the shortest path as it corresponds to \enquote{the cheapest way} for testing $\P$ in our graph, something as simple as Prim's algorithm also works. That is, to prove \cref{thm:NA-ub}, we essentially show that there exists a partition $\P$ separating $a$ and $b$ which is reachable from $\P_0$ through edges of weight at most $O_n(1) \cdot \QC_M(a,b)^2$. Therefore, Prim's algorithm can find such a $\P$, by starting from $\P_0$ and in each step, adding the minimum outgoing edge to the tree. Also, note that while the shortest path algorithm may use fewer queries in many cases, in the worst case, it uses (asymptotically) the same number of queries as Prim's algorithm (\cref{thm:NA-A-gap}).
\section{The Formal Model} \label{sec:model}

We use $\Z^{\geq 0}$ to denote the set of non-negative integers. We use $\tvd(\cdot,\cdot)$ to denote the total variation distance between two distributions. 

We start by reviewing the standard notions of (partially observable) Markov chains. 
A sequence of random variable $X_0,X_1,X_2,\ldots$ with $X_t \in \Omega$ for all $t$ is a {\em Markov chain} if there exists a conditional probability distribution $P_{X|Y}$  such that for every $t\in \Z^{\geq 0}$ we have $X_{t+1}$ is distributed according to $P_{X|Y=X_t}$.\footnote{We note that such Markov chains are called {\em time-invariant} since transition probabilities, conditioned on the current state, are independent of time.} A sequence $Z_0,Z_1,\ldots$ with $Z_t \in \Sigma$ is a {\em partially observable (p.o.) Markov chain} if there exists a Markov chain $X_0,X_1,\ldots$ and a function $O:\Omega \to \Sigma$ such that $Z_t = O(X_t)$ for all $t$. (While a general theory would allow probabilistic functions $O$, we restrict our attention to deterministic functions in this work.) We refer to $\Omega$ as the state space and $\Sigma$ as the observable space. 
Note that a partially observable chain $M$ is given by a triple $(\Omega, P = P_{X|Y},O)$. 

Next, we turn to the central objects of study in this paper, namely Markov chains with rewinding. Before doing so, we briefly touch upon two views of a rewinding strategy. The formal view is simply as a function that maps a history of observations to a non-negative number (which we view as the number of steps we rewind the current time by). An alternate view is that the rewinding strategy builds a tree where, at each time step, the algorithm picks a node of the current tree, and the next time step produces a new child for this node whose state is independent of that of all other nodes, conditioned on the parent's state. The following definition uses the formal view, but we often switch to the tree view in our analyses.

\begin{definition}[\textbf{Markov Chains with Rewinding}]
\label{def:rewinding-strategy}
\label{def:MCR}
For a partially observable Markov chain $M = (\Omega, P,O)$ with observable space
$\Sigma$, a rewinding strategy is a function $A:\Sigma^* \to \Z^{\geq 0}$. A Markov chain $M=(\Omega,P,O)$ 
in initial state $X_0$ with rewinding strategy $A$ generates a sequence of random variables $Z_0,Z_1,Z_2,\ldots$, with companion variables $X_1,X_2,\ldots$ as follows: For every $t \geq 0$, we let $Z_t = O(X_t)$ and $X_{t+1} \sim P_{X|Y = X_{t'}}$ 
where $t' = \max\{0,t-A(Z_0,\ldots,Z_t)\}$. We say that a sequence of random variables $Z_0,Z_1,\ldots,$ is a Markov chain with rewinding if there exists a chain $M$, an initial state $X_0$, and a rewinding strategy $A$ that generates this sequence.
\end{definition} 

Markov chains with rewinding strictly generalize the class of Markov chains, which correspond to {\em passive observation}, i.e., the rewinding strategy $A(\cdot)$ is the constant function $0$, or equivalently, the rewinding tree is a path (and hence ``chain''). Another class of rewinding strategies, widely used in algorithm design, picks some threshold $\tau \in \Z^{\geq 0}$ and sets $A(Z_0,\ldots,Z_t) = t$ if $t = 0 \pmod\tau$ and $A(Z_0,\ldots,Z_t)=0$ otherwise. Equivalently, the rewinding tree consists of paths of length $\tau$ intersecting at the root. Such strategies may be termed {\em resetting strategies}. 

An important subclass of strategies which we study in this work are {\em non-adaptive} rewinding strategies, where $A(Z_0,\ldots,Z_t)= A'(t)$, i.e., the rewinding strategy is a fixed function of the length and does not depend on the actual observations. Note that resettable strategies are special cases of non-adaptive strategies. A simple example of a non-adaptive rewinding strategy inspired by the ``go-with-the-winners'' algorithm \cite{AldousVazirani} is the following: 
$A(Z_0,\ldots,Z_t) = 0$ if $t \not= 0 \pmod\tau$ and $\argmax_\ell\{Z_{t-\ell}\}$ otherwise. That is, the chain-rewinder seeks to maximize the observed variable, and once in every $\tau$ steps, it rewinds to a previous state with the highest observable.

The main property of Markov chains we explore is the {\em identifiability} of the initial state. Specifically, given a p.o.\ Markov chain $M$ and two candidate initial states $a,b \in \Omega$, we ask how long a rewinding strategy has to run (or how many ``queries'' it must make to $M$) to distinguish between these two starting states. The following definition formalizes this notion.
It only considers identifiability from an information-theoretical point of view, i.e.\ the query complexity:

\begin{definition}[\textbf{Query Complexity}] \label{def:QC} \label{def:NAQC}
    Given a p.o.\ Markov chain $M$ and rewinding strategy $A$, let $Z^a_0,Z^a_1,Z^a_2,\ldots$ denote the Markov chain with rewinding obtained from the initial state $X_0 = a$ and let $Z^b_0,Z^b_1,Z^b_2,\ldots$ correspond to the case $X_0 = b$. We say that $(M,A,a,b)$ are $(\epsilon,T)$-identifiable if 
    $\tvd((Z^a_0,\ldots,Z^a_T),(Z^b_0,\ldots,Z^b_T))\geq \epsilon$. The {\em query complexity} of distinguishing $a$ from $b$ in chain $M$ with rewinding strategy $A$, denoted $\QC^A_M(a,b)$, is the minimum $T$ such that $(M,A,a,b)$ are $(1/3,T)$-identifiable. Having this, we define:
    \begin{itemize}[topsep=0pt,itemsep=-10pt]
    \item the (adaptive) query complexity of distinguishing $a$ from $b$ to be: $$\AQC_M(a, b) = \min_A \QC^A_M(a, b),$$
    \item and define the analogous notion for non-adaptive rewinding strategies. Namely,
     \begin{flalign*}
     \NAQC_M(a, b) &= \min_{\text{non-adaptive } A} \QC^A_M(a, b).
     \end{flalign*}
\end{itemize}
\end{definition}

We also consider identifiability from a computational standpoint. That is, we study the time complexity for algorithms that identify the initial state.
Such an algorithm produces the rewinding strategy and queries the Markov chain accordingly. Then, it identifies the initial state based on all the observations.
We call such algorithms \emph{state identification algorithms}.

\begin{definition}[\textbf{State Identification Algorithms}]
\label{def:state-identification-algorithm}
The input of a state identification algorithm consists of a p.o.\ Markov chain $M = (\Omega, P, O)$, two candidate initial states $a, b \in \Omega$, and a hidden initial state $X_0 \in \{a, b\}$.
An (adaptive) state identification algorithm in every step $t \geq 0$, computes the rewinding strategy $A(Z_0, \ldots, Z_t)$, after which $X_{t+1}$ is sampled according to $P_{X|Y=X_{t'}}$ with $t' = \max\{0, t - A(Z_0, \ldots, Z_t)\}$, and the observation $Z_{t+1} = O(X_{t+1})$ is revealed to the algorithm.
A non-adaptive algorithm computes the rewinding strategy before sampling any states.
The algorithm decides the number of steps $T$, and eventually outputs the identified initial state ($a$ or $b$).
We say the algorithm successfully identifies the initial state if the output is correct with probability $2/3$.
\end{definition}

Finally, we define a special subclass that we refer to as {\em canonical} p.o.\ Markov chains.
While general p.o.\ Markov chains are of interest, we state our results in terms of canonical p.o.\ Markov chains for simplicity.  
We show later, in~\cref{thm:canonical}, that canonical chains capture all chains in our context. 
\begin{definition}[\textbf{Canonical p.o.\ Markov Chains}]\label{def:canonical}
We say that p.o.\ Markov chain $M = (\Omega, P,O)$ is a {\em canonical} p.o.\ Markov chain if there exists a state $\sink \in \Omega$ that is a sink state (i.e., $P_{X|Y=\sink} = 1$ if $X = \sink$ and zero otherwise), and $O(a) = 1$ if $a=\sink$ and $0$ otherwise.    
\end{definition} 
Thus, in a canonical p.o.\ Markov chain, the only observable signal is whether the chain has reached the sink or not; and once one reaches the sink, no further observations reveal information. Such chains are the hardest to learn and hence our focus.

\subsection{Partitions}\label{sec:partitions}
In the rest of this section, we introduce the notion of state partitions, i.e.\ partitions of $\Omega$, a key tool both for designing rewinding strategies to distinguish between states, and proving lower bounds on the query complexity of this task. 
Throughout the section we use $n$ as a shorthand for $\card{\Omega}$, $p(a, b)$ for $P_{X|Y=b}(a)$,
and $p(a, C)$ for $\sum_{b \in C} p(a, b)$.
We use $\sink$  to denote the sink state.

\begin{definition}[\textbf{Partition}]  
Given a set $\Omega$, a partition $\P$ is a collection of sets such that each element $a \in \Omega$ appears in exactly one set of the partition.
We often refer to these sets as classes, and use $\P(a)$ to denote the class of $a$ in $\P$. 
\end{definition}

\begin{definition}[\textbf{Total Variation Distance w.r.t.\ Partitions}]
    For two states $a$ and $b$, and a partition $\P$, the total variation distance of $P_{X|Y=a}$ and $P_{X|Y=b}$ w.r.t.\ to $\P$ is
    $$
    d^{\p}_{TV}(a,b):= \frac{1}{2} \sum_{C \in \p} |p(a,C) - p(b,C)| 
    = \max_{A\subseteq \p} \sum_{C \in A} p(b,C) - p(a,C).
    $$ 
    Here, $A$ is a collection of the classes in $\P$. 
\end{definition}

\begin{lemma} \label{lem:distinguish-pair-using-oracle}
    Let $\P$ be a partition of $\Omega$.
    Given two states $a$ and $b$, and oracle access to $\P(x)$ for any state $x \in \Omega$,
    there exists a non-adaptive rewinding strategy that successfully distinguishes between initial states $a$ and $b$ with probability $1 - \epsilon$ using 
    $$
    O\left(\frac{\log 1/\epsilon}{\left(\tvd^\P(a,b)\right)^2} \right)
    $$
    queries.
\end{lemma}
\begin{proof}
    Let $d  = \tvd^\P(a,b)$, and let $A \subseteq \p$ be the collection of classes where the difference in the distributions is maximized:
    $$
    A := \argmax_{A \subseteq \p} \sum_{C \in A} p(b, C) - p(a, C).
    $$
    Note that by definition $p(b, A) - p(a, A) = d$.
    Given a state $x \in \{a, b\}$, we draw $\frac{2 \log 1/\epsilon}{d^2}$ samples of the next step of $x$ (according to distribution $P_{X|Y=x}$), and let $X$ be the fraction of them that land in $A$.
    If the state is $a$, the expected value of $X$ is $p(a, A)$; if the state is $b$, the expected value is $p(b, A)$; and the difference between the two is equal to $d$.
    We declare that the state is $a$ if $X < p(a, A) + \frac{d}{2}$ and that it is $b$ otherwise.
    The probability of error when $x = a$ is equal to:
    $$
    \Pr\left(X \geq p(a, A) + \frac{d}{2}\right) \leq \exp\left(-2 \cdot \frac{2 \log 1/\epsilon}{d^2} \cdot \left(\frac{d}{2}\right)^2 \right) \leq \epsilon.
    $$
    Similarly, the probability of error when $x = b$ is equal to:
    \begin{equation*}
    \Pr\left(X < p(a, A) + \frac{d}{2}\right)
    = \Pr\left(X < p(b, A) - \frac{d}{2}\right)
    \leq \epsilon. \qedhere
    \end{equation*}
\end{proof}


\section{A Polynomially Optimal Non-Adaptive Algorithm}\label{sec:prim}

\subsection{A Non-Adaptive Algorithm with Polynomial Queries}

In this section, we present a non-adaptive algorithm that distinguishes any two (distinguishable) states $a$ and $b$ of an $n$-state canonical Markov chain $M=(\Omega, P)$ (see \cref{def:canonical}) using  ${O_n(1) \cdot \AQC_M(a, b)^{O(n)}}$ queries, where recall $\AQC_M(a, b)$ is the optimal query complexity of distinguishing $a$ from $b$ with a possibly adaptive rewinding strategy, and $O_n(1)$ is used to suppress terms dependent only on $n$. In particular, this implies that the gap between adaptive and non-adaptive query strategies is only polynomial so long as the number of states of the Markov chain is constant. We later show in \cref{sec:gap} that this is the best one can hope for: there are choices of $M, a, b$ where any non-adaptive rewinding strategy requires $\AQC_M(a, b)^{\Omega(n)}$ queries.

Formally, we show:
\begin{theorem}\label{thm:opt-to-the-t}
    Let $M=(\Omega, P)$ be a canonical p.o.\ Markov chain with $n = \card{\Omega}$ states. There is a non-adaptive algorithm (\Cref{alg:opt-to-the-t}) that distinguishes any distinguishable states $a, b \in \Omega$  in time 
    $$O_n(1) \cdot (\AQC_M(a, b)\log \AQC_M(a, b))^{2(n-2)} = O_n(1) \cdot \AQC_M(a, b)^{O(n)}.$$
\end{theorem}

This is equivalent to \cref{thm:NA-ub}.

\subsection{An Informal Overview of the Algorithm}

The algorithm defines a weighted directed graph $G_M=(V, E)$ where each vertex of the graph corresponds to a partitioning of $\Omega$, the state space of the Markov chain. There is an edge from $\mc{P} \in V$ to $\mc{P}' \in V$ iff $\mc{P}'$ is a refinement of $\mc{P}$. Informally speaking, it would be helpful to view each partitioning $\mc{P}$ as a {\em test} that reveals which class of $\mc{P}$ a hidden state belongs to. With this view, we define the weight $w(\mc{P}, \mc{P'})$ of the directed edge $(\mc{P}, \mc{P'}) \in E$ such that it upper bounds the cost of solving test $\mc{P'}$ having free access to test $\mc{P}$. Once we define this graph $G_M$, we compute the shortest path tree on $G_M$ (e.g.\ using Dijkstra's algorithm) starting from the trivial partitioning $\mc{P}_0 = \{ \Omega \setminus \{s\}, \{s\}\}$ that only separates the sink from the other states. Let $\mc{P}_0, \mc{P}_1, \ldots, \mc{P}_k$ be the unique path from $\mc{P}_0$ to a partitioning $\mc{P}_k$. We will define the cost of $\mc{P}_k$ as
$$
    c(\mc{P}_k) = \prod_{i=0}^{k-1} w(\mc{P}_i, \mc{P}_{i+1}).
$$
We prove that each test $\mc{P}$ can indeed be solved with (slightly more than) $c(\mc{P})$ queries. To distinguish state $a$ from $b$, the algorithm takes the partitioning $\mc{P}$ that separates $a$ from $b$ and has the minimum cost $c(\mc{P})$. Our analysis shows that $c(\mc{P})$ can be upper bounded by $O_n(1) \cdot \AQC_M(a, b)^{O(n)}$.

\subsection{The Formal Argument}

\begin{definition}[Source Partition]
The \emph{source partition} $\P_0$ is the trivial partition where the sink is a singleton class and all the remaining states form another class, i.e.\ $\P_0 =\{\Omega \setminus \{\sink\}, \{ \sink\} \}$. 
\end{definition}

For the purposes of this section, we only consider the partitions that separate $s$ from every other state, i.e.\ partitions that refine $\P_0$.

\begin{definition}[Partition Graph] \label{def:partition-graph}
Given a canonical Markov chain $M = (\Omega, P)$, let the partition graph $G_M$ be a weighted directed graph where each node is a partition of the states $\Omega$.
For a partition $\P_1$ and a refinement $\P_2$, there is a directed edge from $\P_1$ to $\P_2$ with weight
$$
w(\P_1, \P_2) = \max_{\substack{a, b: \\ \P_1(a) = \P_1(b) \\ \P_2(a) \neq \P_2(b)}} 1/\left(\tvd^{\P_1}(a, b)\right)^2.
$$
The lengths of the paths are defined multiplicatively. That is, the length of a path is the product of the weights of its edges. For a partition $\p$, its cost $c(\p)$ is equal to the length of the shortest path from $\p_0$ to $\p$.
\end{definition}

Below is a formal overview of the algorithm. The details of using the shortest paths to distinguish between states are given in \cref{lem:using-a-path}.
We use the notion of trees for rewinding. That is, if $X_t$ is sampled by rewinding to $X_{t'}$ (i.e.\ according to distribution $P_{X|Y=X_{t'}}$), then $X_{t'}$ is the parent of $X_t$ in the tree.

\begin{algorithm}[H]
    \caption{A non-adaptive algorithm with polynomial queries.}
    \label{alg:opt-to-the-t}
    \textbf{Input:} A Markov chain $M = (\Omega, P)$, and two states $a$ and $b$

    Let $G_M$ be the partition graph of the Markov chain as in \cref{def:partition-graph}.

    Compute the (multiplicative) shortest path tree on $G_M$ rooted at the trivial partition $\p_0$.

    Let $\p$ be a partition separating $a$ and $b$ that minimizes $c(\p)$. \label{step:find-best-partition}

    Let $\p_0, \p_1, \ldots, \p_k=\p$ be the shortest path to $\p$.

    Let $\epsilon = \Theta_n\left(\frac{1}{c(\p) \log c(\p)}\right)$. 

    Query a tree of height $k$, where the vertices at height $i$ have degree $\Theta_n(\log(1 / \epsilon) \cdot w(\p_{i-1}, \p_i))$. \label{step:begin-using-the-path}

    Use the observations in the tree to identify the class of each vertex at height $i$ in $\p_i$ recursively (see \cref{lem:using-a-path} for the details). \label{step:end-using-the-path}
\end{algorithm}

To analyze the query complexity of \Cref{alg:opt-to-the-t}, we show there exists a path of length \linebreak $O_n(1)\cdot \AQC(a, b)^{2(n-2)}$ to a partition that separates $a$ and $b$ (\cref{lem:short-path-exists}).
Then, we show how this path can be used to distinguish between $a$ and $b$ (\cref{lem:using-a-path}).

First, we prove some auxiliary claims.
The following lemma essentially upper-bounds the weight of the smallest outgoing edge of a partition $\P$, based on two states $a$ and $b$ with $\P(a) = \P(b)$.

\begin{lemma} \label{lem:separating-pairs}
    Take a partition $\P$ and any two states $a, b \in \Omega$ such that $\P(a) = \P(b)$.
    Then, there exists another partition $\P'$ that refines $\P$
    and 
    $$
    w(\P, \P') \leq \left(\frac{n-1}{\tvd^\P(a, b)}\right)^2.
    $$
\end{lemma}
\begin{proof}
Let $G$ be a graph where each node corresponds to a state.
There is an edge between two states $x$ and $y$ if $$\tvd^\p(x, y) < \frac{\tvd^\p(a, b)}{n-1}.$$
We let two states $x$ and $y$ be in the same class of $\p'$ if and only if they are in the same class of $\p$ and in the same connected component of $G$.

First, observe that $a$ and $b$ are separated by $\p'$ because they are in different connected components of $G$.
Assume for the sake of contradiction that they are connected by a path in $G$.
The length of this path is at most $n-1$, and each edge has a weight strictly smaller than $\frac{\tvd^\p(a, b)}{n-1}$.
This is a contradiction since $\tvd^\p$ is a metric, and the total length of the path cannot be smaller than the distance of its endpoints (i.e.\ $a$ and $b$).
Also, if two states are separated in $\p$, then they are separated in $\p'$.
Therefore $\p'$ is a refinement of $\p$.

Second, for any two states $x$ and $y$ such that $\p(x) = \p(y)$ and $\p'(x) \neq \p'(y)$,
it holds that $\tvd^\p(x, y) \geq \frac{\tvd^\p(a, b)}{n-1}$.
Otherwise, there would have been a direct edge between them, and they would have been in the same connected component of $G$, which contradicts $\p'(x) \neq \p'(y)$.
Therefore, by definition we have 
$$w(\p, \p') = \max_{\substack{x, y: \\ \p(x) = \p(y) \\ \p'(x) \neq \p'(y)}} 1/\left(\tvd^{\p}(x, y)\right)^2 \leq 1/\left(\tvd^\p(a, b)/(n-1)\right)^2.$$ This concludes the proof. 
\end{proof}

\begin{corollary} \label{cor:outgoing-edge}
    For a partition $\p$,
    if $\p$ has no outgoing edge of weight at most $h$ in the partition graph,
    then for any two states $a$ and $b$ in the same class of $\p$,
    it holds that 
    $$
    \tvd^\p(a, b) < \frac{n-1}{\sqrt{h}}.
    $$
\end{corollary}
\begin{proof}
Follows directly from \cref{lem:separating-pairs}.
\end{proof}

The following claim states that if there exists a partition $\P$ such that every two states in the same class are closer than $\theta$ w.r.t.\ $\tvd^\P$, then distinguishing any two states in the same class of $\P$ is difficult (roughly speaking, requires $1/\theta$ queries.

\begin{claim} \label{clm:almost-coupling}
    For a partition $\p$ that does not separate $a$ and $b$, let 
    $$\theta(\p) \coloneq \max_{\substack{x,y: \\ \p(x) = \p(y)}} \tvd^\p(x, y).$$
    Then, any (possibly adaptive) rewinding strategy that uses $Q$ queries,
    can successfully distinguish between $a$ and $b$ with probability at most $\frac{1}{2} + \frac{Q \cdot \theta(\p)}{2}$.
\end{claim}
\begin{proof}
Given a rewinding strategy $\mA$ that makes at most $Q$ queries, we consider the query tree of the strategy when the initial state is $a$ and when it is $b$, and present a coupling such that the coupled trees are isomorphic (w.r.t.\ $\p$, meaning the corresponding nodes are in the same class of $\P$) with probability at least $1 - Q \cdot \theta(\p)$.
Then it can be inferred that $\mA$ cannot distinguish between $a$ and $b$ with probability at least $\frac{1}{2}(1 - Q \cdot \theta(\p))$, since the strategy makes the same observation as long as the trees are isomorphic w.r.t.\ $\P_0$. If the strategy is randomized, we fix its random tape and couple the query trees for each possible random tape. That is, we assume without loss of generality that the strategy is deterministic.

Considering two parallel runs of the rewinding strategy where the initial states are $a$ and $b$ respectively, we couple the query trees as follows.
We present this coupling step by step based on what the strategy does. 
In certain cases, we say the coupling has \emph{failed} and couple the rest of the process on each tree independently.
While the coupling has not failed, the two trees are the same w.r.t.\ to $\p$.
That is, they are the same topologically when the states are ignored,
and the states of corresponding nodes in the two trees are in the same class of $\p$.

Let $T_a$ and $T_b$ be the trees obtained in each run, and assume that the coupling has not failed so far. 
Initially, each of the two trees is a single vertex with state $a$ or $b$ respectively.
Observe that $\p(a) = \p(b)$, therefore the trees are indeed the same w.r.t.\ $\p$ in the beginning.
Let $u_a$ be a node of the tree from which the strategy draws a child in the next step on $T_a$.
Then, the corresponding node in $T_b$ will be queried by the strategy in the next step, because the trees are the same w.r.t.\ $\p$, and as a result, the strategy makes the same observations, i.e.\ the sink and none-sink nodes are the same. We also use $u_a$ and $u_b$ to indicate the state of the Markov chain at these nodes.

Since the trees are the same w.r.t.\ $\p$, it holds that $\p(u_a) = \p(u_b)$, therefore, by the definition of $\theta(\P)$, it holds $\tvd^\P(u_a, u_b) \leq \theta(\P)$.
As a result, by the definition of total variation distance, the children drawn from them $u_a$ and $u_b$ can be coupled together such that they are in the same class of $\p$ with probability at least $1 - \tvd^\p(a, b) \geq 1 - \theta(\p)$.
In case they are in the same class, we continue the coupling as is.
Otherwise, we say the coupling has failed, and couple the rest of the process independently for the two trees. The probability of failure at every step is at most $\theta(\p)$. As a result, since the rewinding strategy uses at most $Q$ queries, the overall probability that the coupling fails is at most $Q \cdot \theta(\p)$.

When the coupling does not fail, the rewinding strategy has the same output in the two runs because the two trees are the same w.r.t.\ $\p$.
That is, given the two initial states $a$ and $b$, the output is the same with probability at least $1 - Q \cdot \theta(\p)$.
Therefore, for one of $a$ and $b$ the output is incorrect with probability at least $\frac{1}{2}(1 - Q \cdot \theta(\p))$. This concludes the proof.
\end{proof}

Combining \cref{cor:outgoing-edge,clm:almost-coupling}, for a partition $\P$ with $\P(a) = \P(b)$, we can upper-bound the smallest outgoing edge of $\P$, and prove there exists a short path from $\P_0$ to a partition that separates $a$ and $b$.

\begin{lemma} \label{lem:short-path-exists}
    There exists a path of length $O_n(1)\cdot \AQC(a, b)^{2(n-2)}$ from $\P_0$ to a partition that separates $a$ and $b$.
\end{lemma}
\begin{proof} 
    To prove a path exists, let $S$ be the set of partitions that are reachable from the source partition $\p_0$ using edges of weight at most $h = \left(3(n-1)\AQC(a, b)\right)^2 = O_n(1) \cdot (\AQC(a, b))^2$.
    Take a partition $\p \in S$ that has no outgoing edges inside $S$. Such a partition can be found because an edge $\p_1 \to \p_2$ can exist only if $\p_2$  refines $\p_1$.
    We show that $a$ and $b$ are separated in $\p$, and conclude that the shortest path from the source partition to $\p$ has length at most $h^{n-2}$.

    Assume for the sake of contradiction that $a$ and $b$ are not separated in $\p$.
    By the choice of $\p$, it has no outgoing edges of weight at most $h$.
    By \Cref{cor:outgoing-edge}, it holds that 
    $$\theta(\p) < \frac{n-1}{\sqrt{h}} \leq \frac{1}{3 \AQC(a, b)}. $$
    Since $\p$ does not separate $a$ and $b$, \cref{clm:almost-coupling} implies that any rewinding strategy which uses $\AQC(a, b)$ queries can distinguish between $a$ and $b$ with probability at most
    $$
    \frac{1}{2} + \frac{\AQC(a, b) \cdot \theta(\p)}{2} < \frac{2}{3}.
    $$
    This is a contradiction since by definition there exists a rewinding strategy distinguishing between $a$ and $b$ with $\AQC(a, b)$ queries and success probability $\frac{2}{3}$.

    Finally, observe that $\p$ is reachable from the source partition using at most $n - 2$ edges of weight at most $h$.
    We have already shown there exists a path consisting of edges of weight at most $h$ from the source partition to $\p$.
    The number of edges in this path is at most $n - 2$ because each partition in the path is a refinement of the previous one.
    Hence the length is at most $$h^{n-2} = \left(3(n-1)\AQC(a, b)\right)^{2(n-2)} = O_n(1)\cdot \AQC(a, b)^{2(n-2)}.$$
    This concludes the proof.
\end{proof}

The following lemma describes how a path of length $Q$ can be used to distinguish the states with (slightly more than) $Q$ queries.

\begin{lemma} \label{lem:using-a-path}
    Given a path $\p_0, \p_1, \ldots, \p_k$ of length $Q$, one can identify the class of the initial state in $\p_k$ with $O_n(1) \cdot  Q \log^k Q$ queries,\footnote{recall $k \leq n-2$}
    using a non-adaptive tree of height $k$ where on the $i$-th level the degree of every node is chosen proportionally to $w(\p_{i-1}, \p_i)$ and the class of each node on that level is determined w.r.t.\ $\p_i$.
\end{lemma}
\begin{proof}
    Let $\p_0, \ldots, \p_k$ be a path of length $Q$.
    First,
    assuming that given a state $a$, we have oracle access to $\p_{k-1}(a)$,
    we recall a subroutine that can identify the class of a state in $\p_k$
    by drawing children from it and observing their class in $\p_{k-1}$.
    To do so, we show how to distinguish between any two states $a$ and $b$ that are in different classes of $\p_k$.

    Let $d = \tvd^{\P_{k-1}}(a, b)$
    If $a$ and $b$ are also in different classes of $\p_{k-1}$, then we can distinguish between them simply by accessing their class in $\p_{k-1}$ for free.
    Otherwise, we can invoke \cref{lem:distinguish-pair-using-oracle} to distinguish between $a$ and $b$ by drawing $\frac{2\log 1/\epsilon}{d^2}$ samples of the next step and looking at their class in $\P_{k-1}$.

    Given any state $x$, to identify its class in $\p_k$,
    we perform the above test for every pair $(a, b)$ such that $\p_k(a) \neq \p_k(b)$.
    A state $a$ is called a \enquote{winner} if for every state $b$ with $\P_k(a) \neq \P_k(b)$ the test $(a, b)$ outputs $a$.
    We report that $x$ is in the class $C \in \p_k$ if there exists a winner state in $C$.
    Assuming all the tests involving $x$ were carried out successfully, only  $\P_k(x)$ contains a winner, and $\P_k(x)$ is determined correctly.

    The total number of queries to identify the class of a state in $\p_k$
    is the number of tests ($n^2$) times the number of queries per test ($\frac{2\log 1/\epsilon}{d^2}$).
    Also, it holds by definition:
    $$
        \frac{1}{d^2} \leq w(\p_{k-1}, \p_k).
    $$
    Therefore, the total number of queries is at most
    $$
    2 n^2 \cdot \log (1/\epsilon) \cdot w(\p_{k-1}, \p_k).
    $$

    So far, we have assumed oracle access to the class of a state in $\p_{k-1}$.
    To lift this assumption, we simply perform this procedure recursively.
    That is, to obtain for a state $s$, its class in $\p_{k-1}$, we draw a number of children and look at their classes in $\p_{k-2}$, and so on.
    As a result, we get a tree, where the total number of queries is:
    \begin{align*}
    \prod_{i=1}^k 2 n^2 \cdot \log (1/\epsilon) \cdot w(\p_{i-1}, \p_i)
    & \leq \left(2n^2\log(1/\epsilon)\right)^{k}\prod_{i=1}^k w(\p_{i-1}, \p_i) \\
    & \leq \left(2n^2\log(1/\epsilon)\right)^{k} \cdot Q.
    \end{align*}
    Here, the first inequality follows from rearranging the terms, and the second inequality holds by the definition of $Q$ (the length of the path).

    It remains to set the value of $\epsilon$ such that all tests are carried out successfully with the desired probability of $\frac{2}{3}$.
    Applying the union bound, the probability that at least one of the tests is unsuccessful is at most
    $$
    \epsilon \cdot \left(n^2\log(1/\epsilon)\right)^{k} \cdot Q.
    $$
    By letting $\epsilon = \Theta\left(\frac{1}{n^{2k}Q \log (n^{2k}Q)}\right)$, this probability becomes smaller than $1/3$. Plugging $\epsilon$ back in, the total number of queries this method uses to identify the class of any state in $\p_k$ with success probability $2/3$ is:
    \begin{equation*}
        O\left(\left(2n^2\log(1/\epsilon)\right)^{k} \cdot Q\right)
        = O_n(1) \cdot Q \log^{k} Q. \qedhere
    \end{equation*}
\end{proof}

Finally, we put \cref{lem:short-path-exists,lem:using-a-path} together to derive \cref{thm:opt-to-the-t}.

\begin{proof}[Proof of \cref{thm:opt-to-the-t}]
    By \cref{lem:short-path-exists}, for any two states $a$ and $b$, there exists a path of length at most $\AQC(a, b)^{2(n-2)}$ ending at a partition that separates $a$ and $b$.
    Therefore, the partition $\p$ obtained by computing shortest paths in step \ref{step:find-best-partition} of the algorithm has cost at most $\AQC(a, b)^{2(n-2)}$.
    \cref{lem:using-a-path} shows how the shortest path to $\p$ can be used in steps \ref{step:begin-using-the-path}-\ref{step:end-using-the-path}, to identify the class of the initial state in $\p$ using $O_n(1) \cdot (\AQC(a, b)\log \AQC(a, b))^{2(n-2)}$ queries, hence distinguishing between $a$ and $b$.

    Regarding the runtime, observe that the strategy is implicitly (non-adaptively) computed in $O_n(1)$.
    The number of state partitions (number of vertices in $G_M$) is at most $n^n = O_n(1)$.
    For an edge $\P \to \P'$, the weight can be computed in time $O(n^3)$ by iterating over every relevant pair of states and computing the distance $\tvd^\P$ in $O(n)$.
    Then the shortest path tree can be computed using a polynomial-time algorithm (e.g.\ Dijkstra's).
    After the shortest path is computed, the degree on each level of the query tree is determined.
    For each node of height $i$, its calls in $\P_i$ can be determined in time $O(2n^2 \cdot \log (1/\epsilon) \cdot w(\p_{k-1}, \p_k))$, i.e.\ proportional to the degree.
    Therefore, the run time of the algorithm is $O_n(1) \cdot (\AQC(a, b)\log \AQC(a, b))^{2(n-2)} = O_n(1) \cdot (\QC(a, b))^{O(n)}$.
\end{proof}

\section{The Polynomial Gap of Adaptivity}
\label{sec:gap}

In this section, we prove the following theorem which directly implies \cref{thm:NA-A-gap}.

\begin{theorem}\label{thm:gap}
    For any parameters $n, d \geq 1$, there exists a canonical Markov chain $M=(\Omega, P)$ with $n = |\Omega|$ states and two states $a, b \in \Omega$ such that 
    $$
    \AQC_M(a, b) = O(n^2 d), \qquad\qquad
    \NAQC_M(a, b) = \Omega(d^{(1-o(1))n}).
    $$ 
\end{theorem}
    We remark that the optimal rewinding strategy, which distinguishes between $a$ and $b$ using $O(n^2d)$ queries, can be trivially turned into an algorithm with the same run time. Therefore, \cref{thm:gap} implies the same gap between the optimal adaptive and non-adaptive run times of state identification algorithms. As such, in this section, we abuse the notation and use the terms rewinding strategy and algorithm interchangeably.

    We prove the claim given $d>0$ and, the Markov chain $M$ appeared in \cref{fig:Gap-A-NA}. 
    In this Markov chain, the states form a directed path toward the sink state, with states labeled as $q_1,..., q_{n-2}$ from left to right.
    We label the last state as $D$. 
    For all $i \in [n-3]$, state $q_i$ goes to $q_{i+1}$ with probability $\frac{1}{2d}$.
    The state $q_{n-2}$ goes to the sink state $s$ with the same probability.
    Each state in $q_1,..., q_{n-2}$ goes to itself with probability $\left(1-\frac{1}{d}\right)/2$.
    All $q_1,..., q_{n-2}$ go to the state $D$ with probability $\frac{1}{2}$, while state $D$ goes to the sink with probability 1.

     We prove that there exists an adaptive algorithm distinguishing $q_1$ and $q_2$ with $O(n^2d)$ queries, while any non-adaptive algorithm needs at least $\Omega(d^{(1-o(1))n})$ queries. 
     The intuition is as follows. 
     An adaptive algorithm can easily test whether each new state opened is the $D$ state or not by simply doing a test with constant queries, and opening another child instead to continue its path. 
     This is not possible for non-adaptive algorithms that come across the $D$ state with probability $\frac{1}{2}$, and when this happens the algorithm intuitively loses any information on the state it started from. 
     So it needs a lot of paths to account for this loss of information. 

\subsection{The Adaptive Algorithm}
\begin{lemma}\label{lem:a-gap}
    There exists an adaptive algorithm that distinguishes $q_1$ and $q_2$ in Markov chain $M$ of \cref{fig:Gap-A-NA} with $O(n^2d)$ queries.
\end{lemma}

\begin{proof}
    To describe the algorithm, we first introduce a boolean $D$-test that can distinguish the state $D$ from any other state. Given an unknown state $a$, we take a constant number of children of $a$. If all of these children are the sink, then $a$ is the $D$ state; otherwise, $a$ is not the $D$ state. Note that the algorithm encounters state $D$ in only an $\epsilon$ fraction of its queries.

    Now we describe the adaptive algorithm $\mathcal{A}$ as follows:
    \begin{itemize}
        \item Given an initial state $a \in \{q_1, q_2\}$, open a child $c$ of $a$.
        \item If $c$ is not a sink, run the $D$-test on $c$. 
            \begin{itemize}
                \item If the $D$-test is positive (indicating $c$ is $D$), then open another child from the parent of $c$.
                \item If the $D$-test is negative, continue to open a child from $c$.
            \end{itemize}
    \end{itemize}
    
    This algorithm will produce a path from $a$ to the sink that avoids state $D$. The length of this path helps determine whether the algorithm started at $q_1$ or $q_2$. Let $E_1$ and $E_2$ represent the expected lengths of such paths when starting from $q_1$ and $q_2$, respectively, and let $\text{Var}_1$ and $\text{Var}_2$ denote the variances of these path lengths. Note that the paths do not include encounters with state $D$ or additional queries from the $D$-tests. 

    Since each transition has a probability of $ \frac{1}{2d}$ of moving to the next state in the sequence, 
    \[
    E_1 = (n-2) \cdot \frac{d}{2} \quad \text{and} \quad E_2 = (n-3) \cdot \frac{d}{2}.
    \]
    
    The variance of a path length in both cases, $\text{Var}_1$ and $\text{Var}_2$, is bounded by the variance of the sum of $(n-2)$ or $(n-3)$ independent geometric random variables with success probability $ \frac{1}{2d}$. We have
    \[
    \text{Var}_1 = \frac{2(n-2)(1-1/d)}{(1/d^2)} \leq 2nd^2 \quad \text{and} \quad \text{Var}_2 = 2\frac{(n-3)(1-1/d)}{(1/d^2)} \leq 2nd^2.
    \]

    Let $X_1, \dots, X_k$ denote the lengths of the $n$ paths generated by $\mathcal{A}$. To distinguish between $q_1$ and $q_2$, we compare the sample mean of path lengths, $\overline{X} = \frac{1}{k} \sum_{i=1}^k X_i$, to $E_1$ and $E_2$. 

    Let $\sigma$ be the standard deviation of random variable $X$. We have $\sigma^2 \leq nd^2/k$.
    By applying Chebyshev’s Inequality to the sample mean $\overline{X}$, when the starting state is $q_1$, we have
    \[
    \Pr\left( |\overline{X} - E_1| \geq d/2 \right) \leq  \frac{4\sigma^2}{d^2}\leq \frac{8n}{k}.
    \]
    Similarly for the case that the starting state is $q_2$ we have
    \[
    \Pr\left( |\overline{X} - E_2| \geq d/2 \right) \leq  \frac{4\sigma^2}{d^2}\leq \frac{8n}{k}.
    \]

    To  distinguish between the two expected path lengths $E_1$ and $E_2$ with high probability, we want the sample mean $\overline{X}$ to be close to $E_1$ or $E_2$ depending on the starting state. Thus, we set $k = 100n$. In the case that the starting state is $q_1$ with probability $2/3$ we have $|\overline{X} - E_1| \leq d/2$, and if the starting state is $q_2$ with high probability we have $|\overline{X} - E_1| \leq d/2$. Therefore, the algorithm succeeds with probability $2/3$.

    The total size of the tree that this algorithm opens is $O(nd \cdot n) = O(n^2 d)$, and with probability at least $\frac{2}{3}$, this algorithm will distinguish between $q_1$ and $q_2$.
\end{proof}

\subsection{The Non-Adaptive Lower Bound}

\begin{lemma}\label{lem:na-gap}
    Any non-adaptive algorithm distinguishing $q_1$ and $q_2$ in Markov chain $M$ of \cref{fig:Gap-A-NA} needs $\Omega(d^{(1-o(1))n})$ queries.
\end{lemma}

\begin{proof}
Assume there exists a non-adaptive algorithm $\mathcal{A}$ designed to distinguish the states $q_1$ and $q_2$ in the Markov chain $M$.
Given either $q_1$ or $q_2$ as the initial state, $\mathcal{A}$ generates a tree $T$ that correspond to states of the Markov chain. 
For both starting states $q_1$ and $q_2$, the tree $T$ has the same structure due to the non-adaptive nature of $\mathcal{A}$.

We will construct a coupling of the execution of $\mathcal{A}$ when starting from $q_1$ versus starting from $q_2$, to show that the algorithm needs a large number of queries to successfully distinguish between the two initial states.

Let $T$ denote the tree generated by $\mathcal{A}$, with each node in $T$ representing a queried state in the Markov chain. 
We define the two coupled versions of the tree as $T_{q_1}$ starting at $q_1$, and $T_{q_2}$, starting at $q_2$.
We define the coupling as follows. 
If $T_{q_1}$ remains in its current state, then $T_{q_2}$ also self-loops in its current state.
If $T_{q_1}$ transitions to the next state to the right (from $q_i$ to $q_{i+1}$), then $T_{q_2}$ makes the same transition to the right.
If $T_{q_1}$ transitions to the dummy state $D$, then $T_{q_2}$ also transitions to $D$.

If the sink $s$ is reached in either tree, 
but not the other, 
we say the trees have decoupled, 
as this reveals information about the starting state.  
Since we coupled reaching to the state $D$, 
the decoupling happens when in $T_{q_2}$ there is a transition from $q_{n-2}$ to the $s$ and 
 $T_{q_1}$ transitions from $q_{t-3}$ to $q_{n-2}$. 

Let $p_{\text{decouple}}$ denote the probability that the trees $T_{q_1}$ and $T_{q_2}$ decouple within a path of length $k$. Specifically, $p_{\text{decouple}}$ is the probability that one of the trees reaches $s$ while the other does not:
$$
p_{\text{decouple}} = \Pr(\text{reach } s \text{ in } T_{q_2} \text{ but not in } T_{q_1}).
$$

By the structure of $M$ the probability that using a path of length $k$ starting from $q_2$, we reach sink is 
$$
p_{\text{decouple}} \leq \left(\frac{1}{2}\right)^k  \min \left( \binom{k}{n-2} \left(  \frac{1}{d} \right)^{n-2} , 1\right).
$$
We prove that $p_{\text{decouple}} \leq \frac{1}{d^{(1-o(1))n}}$, 
in the case that $k > 2n\log(d)$, 
since $\min ( \binom{k}{n-2} \left(  \frac{1}{d} \right)^{n-2} , 1) \leq 1$, we have 
$$p_{\text{decouple}} \leq (\frac{1}{2})^k \leq e^{-n\log(d)} \leq \left(\frac{1}{d}\right)^n.$$
In the case that $k \leq 2n\log(d)$, 
$$p_{\text{decouple}} \leq \binom{k}{n-2} \left(\frac{1}{d} \right)^{n-2} \leq \left(\frac{ek}{n}\right)^n\left(\frac{1}{d} \right)^{n-2} \leq \left(2e \log(d)\right)^{n} \left(\frac{1}{d} \right)^{n-2}\leq \left(\frac{1}{d}\right)^{(1-o(1)) n}.$$

Now since, $p_{\text{decouple}} \leq \frac{1}{d^{(1-o(1))n}}$, using union bound on all the paths in tree $T$,the total probability that decoupling happens is at most $\frac{|T|}{d^{(1-o(1))n}}$, if $|T| < d^{(1-o(1))n}/3$, then the coupling fails with probability at most $1/3$ and the algorithm $\mathcal{A}$ distinguishes $q_1$ and $q_2$ with probability at most $1/3$. Therefore, any non-adaptive algorithm that distinguishes $q_1$ and $q_2$ should use a tree of size $\Omega(d^{(1-o(1))n})$.
\end{proof}

\begin{proof}[Proof of \cref{thm:gap}]
    Follows from combining \cref{lem:na-gap,lem:a-gap}.
\end{proof}

\section{Generality of Canonical Markov Chains}\label{sec:generality}

In this section, we show that identifying the initial state of a canonical Markov chain (\cref{def:canonical}) is (essentially) as hard as the same task in any partially observable Markov chain via a reduction. That takes into account both the query complexity of the rewinding strategy and the runtime of the algorithm. 
For a state identification algorithm $\mathcal{A}$ we write $\mathrm{T}(\mathcal{A})$ for the time it spends outside the oracle that samples the next state.

In the generalized version, a Markov chain $M = (\Omega, P)$ is complemented by an observation function $O: \Omega \to \Sigma$, i.e.\ every state $x \in \Omega$ in the Markov chain outputs an observation $O(x) \in \Sigma$.
The states remain hidden, but in every step, the algorithm can see the observation and use it to identify the initial states.
A canonical Markov chain can be expressed in this generalized version by letting the set of possible observations be $\Sigma = \{0, 1\}$, and letting $O(x) = \mathbf{1}(x=s)$ where $s$ is the sink state. That is, the observation made at any state is whether it is the sink state.
First, We show that any adaptive strategy taking $Q$ queries in the generalized model can be emulated in the canonical Markov chain using $O(\card{\Sigma}Q)$ queries. Secondly, we show that for non-adaptive algorithms, the gap between the query complexity of canonical and generalized Markov chains is at most near-quadratic. 

Formally, we define a reduction from a generalized Markov chain $M$ to a canonical Markov chain $\Mh$ below. See \cref{fig:gen,fig:gen-red} for an example of this reduction.

\begin{definition}[\textbf{Reduction to Canonical Markov chains}]

     Let $M = (\Omega, P)$ be a Markov chain with state space $\Omega$, transition matrix $P$, and an observation function $O: \Omega \to \Sigma$ that maps each state to an observable output. Let the emulated Markov chain of $M$ be a canonical Markov chain $\Mh = (\hat{\Omega}, \Ph)$ with state space $\hat{\Omega}$ and transition matrix $\Ph$ as follows:

    For each state $x \in \Omega$, there exists a unique corresponding state $\hat{x} = \phi(x)$ in $\hat{\Omega}$, where $\phi: \Omega \to \hat{\Omega}$ is an injective mapping. Let $k = \card{\Sigma}$ and and $\Sigma = \{\sigma_1, \ldots, \sigma_k\}$. Given a parameter $0<q<1$ we have
    \begin{enumerate}
        \item For every pair of states $x, x' \in \Omega$ with transition probability $p_{x, x'}$ in $P$, we introduce an extended path in $\Mh$ consisting of a series of intermediate states $d^{x,x'}_1, d^{x,x'}_2, \ldots, d^{x,x'}_{k-1}$ with transition probability 1 between consecutive states, such that $\hat{x}$ transitions to $\hat{x}'$ via these intermediate states. $\hat{x}$ transitions to $d^{x,x'}_1$ with probability $qp_{x, x'}$ and $d^{x,x'}_{k-1}$ transitions to $\hat{x}'$ with probability $1$.
        \item To emulate the observation function $O$, we introduce a set of \emph{special states} $ \sigma_1, \sigma_2, \ldots, \sigma_{k} =s$ in $\Mh$, each representing a unique observation in $\Sigma$. $s$ is the sink state of the canonical Markov chain $\Mh$. Each state $\hat{x} \in \hat{\Omega}$ in $\Mh$ transitions to a special state corresponding to its observation in $M$ with probability $1-q$. Special state $\sigma_i$ transitions to $\sigma_{i+1}$ with probability $1$  for $i \in [k-1]$.
    \end{enumerate}
\end{definition}

\begin{figure}[h]
    \centering
    \scalebox{0.8}{%
        \begin{tikzpicture}[->, >=stealth', shorten >=1pt, auto, node distance=2.5cm, scale=1, transform shape]

    \node[state] (s1) {$s_1$};
    \node[state] (s2) [below left of=s1] {$s_2$};
    \node[state] (s3) [below right of=s1] {$s_3$};

    \path[every loop/.style={min distance=10mm, looseness=10}]
        (s1) edge[bend right] node[pos=0.5, left] {$p_{12}$} (s2)
        
        (s2) edge[bend right] node[pos=0.5, left] {$p_{21}$} (s1)
        (s2) edge[bend right] node[pos=0.5, below] {$p_{23}$} (s3)
        (s3) edge[bend right] node[pos=0.5, right] {$p_{31}$} (s1);

\end{tikzpicture}
    }
    \caption{Represents a general Markov chain $M$ with observations $\Sigma = \{\sigma_1, \sigma_2, \sigma_3\} $, where $O_{s_1} = \sigma_1, O_{s_2} = \sigma_2$, and $O_{s_3} = \sigma_2$. }
    \label{fig:gen}
\end{figure}
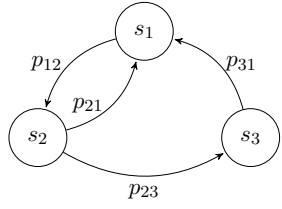

\begin{figure}[h]
    \centering
    \scalebox{0.8}{%
        \begin{tikzpicture}[->, >=stealth', shorten >=1pt, auto, node distance=3cm, scale=1, transform shape]

    \tikzstyle{every state}=[ fill={rgb,255:red,140; green,200; blue,220}, draw=black, text=black, minimum size=1.25cm]
    \tikzstyle{dummy}=[fill=none, draw=black, text=black, minimum size=0.4cm, circle]
    \tikzstyle{sink}=[rectangle, fill=gray!30, draw=black, minimum size=0.8cm, text centered]

    \node[state] (s1) {$
    \hat{s}_1$};
    \node[dummy] (d1) [left=1cm of s1] {$d^{s_1,s_2}_1$};
    \node[dummy] (d2) [below left=1cm  of d1] {$d^{s_1,s_2}_2$};
    \node[state] (s2) [below=1cm of d2] {$\hat{s}_2$};

    \node[dummy] (d3) [right=1cm of s1] {$d^{s_1,s_3}_1$};
    \node[dummy] (d4) [below right=1cm of d3] {$d^{s_1,s_3}_2$};
    \node[state] (s3) [below=1cm of d4] {$\hat{s}_3$};

    \node[dummy] (d5) [right = 1cm of s2] {$d^{s_2,s_1}_1$};
    \node[dummy] (d6) [above right=1cm  of d5] {$d^{s_2,s_1}_2$};

    \node[dummy] (d7) [below right =10mm of s2] {$d^{s_2,s_3}_1$};
    \node[dummy] (d8) [below left=1cm of s3] {$d^{s_2,s_3}_2$};

    \node[dummy] (si1) [right=1cm of d4] {$\sigma_1$};
    \node[dummy] (si2) [right=1cm of si1] {$\sigma_2$};
    \node[sink] (sink) [right=1cm of si2] {$s$};

    \path
        (s1) edge node[pos=0.5, above]{$qp_{12}$} (d1)
        (d1) edge node[pos=0.5, above]{$1$}(d2)
        (d2) edge node[pos=0.5, above left]{$1$}(s2)
        (s3) edge node[pos=0.5, right]{$qp_{31}$}(d4)
        (d4) edge node[pos=0.5, above right]{$1$}(d3)
        (d3) edge node[pos=0.5, above]{$1$}(s1)
        (s2) edge node[pos=0.5, above]{$qp_{23}$}(d5)
        (d5) edge node[pos=0.5,  below right]{$1$}(d6)
        (d6) edge node[pos=0.5, right]{$1$}(s1)
        (s2) edge node[pos=0.5, below left]{$qp_{13}$} (d7)
        (d7) edge node[pos=0.5, below]{$1$}(d8)
        (d8) edge node[pos=0.5, below right]{$1$}(s3)
        (si1) edge node[pos=0.5, below]{$1$}(si2)
        (si2) edge node[pos=0.5, below]{$1$}(sink)

        (s1) edge[bend left,gray] node[pos=0.5, right]{$1-q$}(si1)
        (s2) edge[bend left,gray] node[pos=0.5, below]{$1-q$}(si2)
         (s3) edge[bend right,gray] node[pos=0.5, below]{$1-q$}(si2);

\end{tikzpicture}
    }
    \caption{Represents a canonical Markov chain $\Mh$, with parameter $0<q<1$, where there is an injection $\phi$ that reduces distinguishing states $a$ and $a'$ in $M$ to distinguishing $\phi(a)$ and $\phi(a')$ in $\Mh$.}
    \label{fig:gen-red}
    
\end{figure}
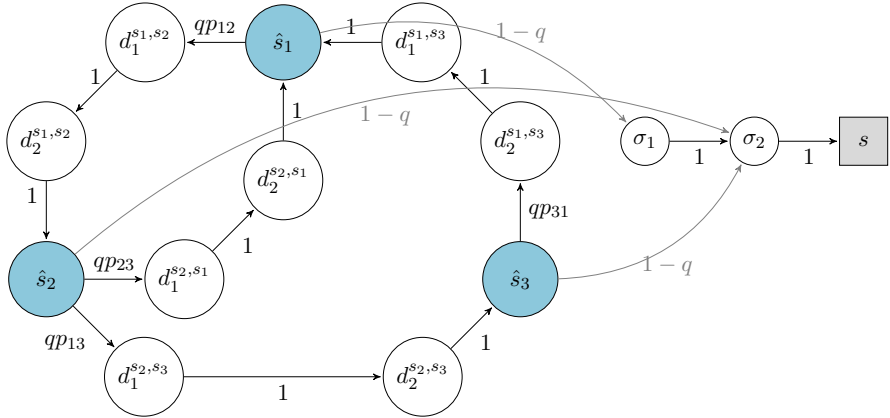

\subsection{Adaptive Query Complexity}
\begin{theorem}\label{thm:canonical}
    Let $M = (\Omega, P)$ be a Markov chain in the generalized model with an observation function $O: \Omega \to \Sigma$,
    Then, there exists a canonical Markov chain $\Mh = (\hat{\Omega}, \Ph)$  states and an injective mapping of the states $\phi: \Omega \to \hat{\Omega}$, 
    such that for any two states $a, a' \in \Omega$, 
    $$\AQC_\Mh(\phi(a), \phi(a')) = \Theta \big(\card{\Sigma}\AQC_M(a, a')\big).$$
\end{theorem}

\begin{proof}
    Let Markov chain $\Mh = (\hat{\Omega}, \Ph)$, be the emulation of $M$ with parameter $q = 1/2$.

    Now we show $$\AQC_\Mh(\ah, \ah') \leq O\big(\card{\Sigma} \AQC_M(a, a')\big),$$ where $a, a' \in S$, $\ah = \phi(a)$, and $\ah' = \phi(a')$.
    Given an algorithm $\mA$ that distinguishes between $a$ and $a'$ between $Q = \AQC_M(a, a')$ queries, we present an adaptive algorithm $\mAh$ in $\Mh$ using $O\big(\card{\Sigma} Q\big)$ queries as follows.
    
    Each node in the query tree of $\mA$ has a corresponding node in the query tree of $\mAh$.
    Initially, the roots correspond to each other,
    and if the root of $\mA$ is in state $x\in \Omega$,
    then we let the root of $\mAh$ be in state $\hat{x} = \phi(x)$.
    Whenever a new state is drawn by $\mA$, we need to emulate the observation.
    To do so, $\mAh$ repeatedly draws paths of length at most $\card{\Sigma}$ until one of them reaches the sink.
    The number of paths is constant in expectation since we rewind a special state on the first step of the path with probability $\frac{1}{2}$.
    The observation can be inferred afterward since it is equal to the length of the path that reaches the sink.
    With the observation in hand, $\mA$ would choose another node to draw a child from.
    To emulate this, $\mAh$ repeatedly draws paths of length at most $\card{\Sigma}$ until one of them \emph{does not} reach the sink.
    The end of the path corresponds to the child node.
    
    Observe that starting at a node with state $\hat{x}$ (corresponding to $x$),
    the probability that the node at the end of the path is in $\hat{x}'$ (corresponding to $x'$) is $p_{xx'}$ because we are conditioning on the fact that the path did not move to a special state on the first step. The expected number of these paths is also constant. Therefore, each query of $\mA$ can be emulated by $O(\card{\Sigma})$ queries of $\mAh$.

    To complete the proof, we aim to show that 
$$
\Omega\big(\card{\Sigma} \cdot \AQC_M(a, a')\big) \leq \AQC_\Mh(\phi(a), \phi(a')).
$$

Let $\mAh$ be a rewinding algorithm on $\Mh$ that starts at $\phi(a)$ or $\phi(a')$ and distinguishes these two states. Let $t'$ be the number of queries made by $\mAh$, and let $X'_0, \dots, X'_{t'}$ be the observed states in $\Mh$ with corresponding observations $Z'_0, \dots, Z'_{t'}$, where for any $i \in [t']$, $Z'_i \in \{0,1\}$. If $Z'_i = 1$, then $X'_i$ is a sink state.

    Let us define $L'_i$ as a set of observed states that are in distance $i$ from the root of the query tree of rewinding algorithm $\mA$. More formally we have $L'_0 = \{X'_0\}$, and for any $i>0$, let $L'_i$ be the set of observed states that are a child of any states in $L'_{i-1}$.  

    Now we aim to build a rewinding algorithm $\mA$ using $\mAh$.
    Let $X_0 = X'_0 $ be the starting state of the algorithm For any maximal disjoint subtree $T'$ that is rooted at $L'_1$ and $T' \subset L'_1 \cup \dots \cup L'_{\card{\Sigma}}$, 
    if there exists a state in $T'$ with observation $1$, or the depth of $T$ is less than $\card{\Sigma}$
    we do nothing, otherwise we rewind a child of $X_0$.
    Similarly, let $L_i$ denote the states rewinded in $\mA$, we set $L_0 = \{X_0\}$ and $L_1$ as the set of rewinded children of $X_0$.

Now we build $L_i$ inductively such that there is a function $S$ that for any state $X \in L_i$, we have $S(X)$ as a set of non-sink states in $L'_{i \cdot \card{\Sigma}}$, and each non-sink states in $L'_{i \cdot \card{\Sigma}}$ is covered by a state in $L_i$.

    The base Case is $i = 0$.
     For any $i > 0$, given any state $X \in L_i$ and $S(X)$ we define $L_{i+1}$ as follows.
     Consider the maximal disjoint subtrees $T$ that is rooted in a state $r$ in $L'_{i \cdot \card{\Sigma} + 1}$ and $r$ is a child of a state in $S(X)$. 
     Also let $T$ be a sub set of $L'_{i \cdot \card{\Sigma} + 1} \cup \dots \cup L'_{(i+1) \cdot \card{\Sigma}}$. 
     Let the set of all such subtrees $T_i$. For each subtree $T \in T_i$,
     if any state in the subtree has an observation $1$ (i.e., it is a sink state), or if the subtree has depth less than $\card{\Sigma}$, do nothing.
     Otherwise, rewind a child of $X$ in $\mA$ corresponding to the subtree, and include this child in $L_{i+1}$.

    Now we prove $\mA$ can distinguish between $a$ and $a'$. First, note that any state $y\in L'_{ik}$, is either the sink or there exists a state $x  \in \Omega$ where $\phi(x) = y$. Since the length of transition paths for any two states $\phi(x)$ and $\phi(x')$ is $k$ for any $x,x' \in \Omega$.  Also, since the transition paths are unique, any subtree $T\in T_i$ has the same states in one level of tree. 
    Therefore, any such tree reveals a transition from $x$ to $x'$ if it starts at a child of $\phi(x)$ and at least  one path in $T$ ends at $\phi(x')$. 
    
    In the case that $T$ reaches to the sink, $\mA$ simply uses the observation of the last state. Therefore, if $\mAh$ can distinguish $\phi(a)$ and $\phi(a')$, $\mA$ can distinguish between $a$ and $a'$.  
    Note that the size of any subtree of depth $\card{\Sigma}$ is at least $\card{\Sigma}$. Since for any rewinded state in $\mA$ there exists  subtree of depth $\card{\Sigma}$ in $\mAh$, $\mA$ can distinguish $a$ and $a'$ in $\frac{\AQC_\Mh(\phi(a), \phi(a'))}{\card{\Sigma}}$. 
\end{proof}

\begin{remark}
        This reduction preserves running time up to the same multiplicative overhead: for any adaptive state identification algorithm $\mathcal{A}$ for $M$ that makes $Q$ queries and runs in time $Time(\mathcal{A})$, there is an algorithm $\mAh$ for $\Mh$ that runs in time $Time(\mathcal{A}) + O(kQ)$, where $k=\card{\Sigma}$. Also, any adaptive algorithm $\mAh$ for $\Mh$ with $Q'$ queries and time $Time(\mAh)$ yields an algorithm for $M$ with $\Omega(Q'/k)$ queries and time $Time(\mAh) + O(Q')$.
\end{remark}

\subsection{Non-Adaptive Query Complexity}

\begin{theorem}\label{thm:NA-canonical}
    Let $M = (\Omega, P)$ be a Markov chain in the generalized model with an observation function $O: \Omega \to \Sigma$,
    Then, there exists a canonical Markov chain $\Mh = (\hat{\Omega}, \Ph)$  states and an injective mapping of the states $\phi: \Omega \to \hat{\Omega}$, 
    such that for any two states $a, a' \in \Omega$, 
    $$\Omega\big(\card{\Sigma}\NAQC_M(a, a')\big) \leq  \NAQC_\Mh(\phi(a), \phi(a')) \leq \Tilde{O}\big(\card{\Sigma}\NAQC_M(a, a')^2\big).$$
\end{theorem}

\begin{proof}
     Consider a non-adaptive algorithm $\mathcal{A}$ operating on a Markov chain $M$ to distinguish between the states $a$ and $a'$. Given a fixed constant $c_1>0$ and $Q$ as the query complexity of $\mathcal{A}$, let Markov chain $\Mh = (\hat{\Omega}, \Ph)$, be the emulation of $M$ with parameter $q = 1-c_1Q$.

     First, we prove  
     $$\NAQC_\Mh(\phi(a), \phi(a')) \leq \Tilde{O}\big(\card{\Sigma} Q^2\big).$$

We build a non-adaptive algorithm $\mAh$ on $\Mh$ that distinguishes $\phi(a)$ and $\phi(a')$ with $\Tilde{O}(\card{\Sigma}Q^2)$ queries. For any rewinded state $X$ to $Y$ in $\mA$ and a fixed constant $c_2> c_1$, we rewind $c_2Q\log(Q)$ paths of size $\card{\Sigma}$. 
We pick one of such paths at random to rewind from next, where the last state of the picked random path represents $\hat{Y}$. Now, at any point that $\mA$ rewinds $Y$, $\mAh$ will rewind $\hat{Y}$.
Therefore for any query in $\mA$, $\mAh$ produces $c_2\card{\Sigma}Q\log(Q)$ queries. To get to the observation of each state in $\hat{\Omega}$, $\mAh$ needs to query such paths such that at least one of the paths gets to the special states. The probability that there exists a state queried in $\mA$ whose random paths in $\mAh$ do not include a special state is 
    $$Q(1- \frac{1}{c_1Q})^{c_2Q\log(Q)} \leq Q\exp{(-\frac{c_2Q\log(Q)}{c_1Q})}= Q^{1-c_2/c_1}.$$

The probability that among the $Q$ picked random paths at least one of the random paths is the special path is $(\frac{1}{c_1Q})Q = \frac{1}{c_1}.$
By setting $c_2>c_1>0$ such that $\frac{1}{c_1}, Q^{1-c_2/c_1} < 1/10$, the algorithm distinguishes $\phi(a)$ and $\phi(a')$ with high probability.

  To complete the proof, we show that 
$$
\Omega\big(\card{\Sigma} \cdot \NAQC_M(a, a')\big) \leq \NAQC_\Mh(\phi(a), \phi(a')).
$$

The proof is similar to the lower bound of reduction for the adaptive algorithm. However, we state the proof for the sake of completeness.

Let $\mAh$ be a rewinding algorithm on $\Mh$ that starts at $\phi(a)$ or $\phi(a')$ and distinguishes these two states. Let $t'$ be the number of queries made by $\mAh$, and let $X'_0, \dots, X'_{t'}$ be the observed states in $\Mh$.

    Let us define $L'_i$ as a set of observed states that are in distance $i$ from the root of the query tree of rewinding algorithm $\mA$. More formally we have $L'_0 = \{X'_0\}$, and for any $i>0$, let $L'_i$ be the set of observed states that are a child of any states in $L'_{i-1}$.  

    Now we aim to build a rewinding algorithm $\mA$ using $\mAh$.
    Let $X_0 = X'_0 $ be the starting state of the algorithm For any maximal disjoint subtree $T'$ that is rooted at $L'_1$ and $T' \subset L'_1 \cup \dots \cup L'_{\card{\Sigma}}$, 
    if the depth of $T$ is less than $\card{\Sigma}$
    we do nothing, otherwise we rewind a child of $X_0$.
    Similarly, let $L_i$ denote the states rewinded in $\mA$, we set $L_0 = \{X_0\}$ and $L_1$ as the set of rewinded children of $X_0$.

Now we build $L_i$ inductively such that there is a function $S$ that for any state $X \in L_i$, we have $S(X)$ as a set of non-sink states in $L'_{i \cdot \card{\Sigma}}$, and each non-sink states in $L'_{i \cdot \card{\Sigma}}$ is covered by a state in $L_i$.

    The base Case is $i = 0$.
     For any $i > 0$, given any state $X \in L_i$ and $S(X)$ we define $L_{i+1}$ as follows.
     Consider the maximal disjoint subtrees $T$ that is rooted in a state $r$ in $L'_{i \cdot \card{\Sigma} + 1}$ and $r$ is a child of a state in $S(X)$. 
     Also let $T$ be a sub set of $L'_{i \cdot \card{\Sigma} + 1} \cup \dots \cup L'_{(i+1) \cdot \card{\Sigma}}$. 
     Let the set of all such subtrees $T_i$. For each subtree $T \in T_i$,
      if the subtree has depth less than $\card{\Sigma}$, do nothing.
     Otherwise, rewind a child of $X$ in $\mA$ corresponding to the subtree, and include this child in $L_{i+1}$.

    Now we prove $\mA$ can distinguish between $a$ and $a'$. First, note that any state $y\in L'_{ik}$, is either the sink or there exists a state $x  \in \Omega$ where $\phi(x) = y$. Since the length of transition paths for any two states $\phi(x)$ and $\phi(x')$ is $k$ for any $x,x' \in \Omega$.  Also, since the transition paths are unique, any subtree $T\in T_i$ has the same states in one level of the tree. 
    Therefore, any such tree reveals a transition from $x$ to $x'$ if it starts at a child of $\phi(x)$ and at least one path in $T$ ends at $\phi(x')$. 
    
    In the case that $T$ reaches the sink, $\mA$ simply uses the observation of the last state. Therefore, if $\mAh$ can distinguish $\phi(a)$ and $\phi(a')$, $\mA$ can distinguish between $a$ and $a'$.  
    Note that the size of any subtree of depth $\card{\Sigma}$ is at least $\card{\Sigma}$. Since for any rewinded state in $\mA$ there exists  subtree of depth $\card{\Sigma}$ in $\mAh$, $\mA$ can distinguish $a$ and $a'$ in $\frac{\NAQC_\Mh(\phi(a), \phi(a'))}{\card{\Sigma}}$. 
\end{proof}

\begin{remark}
        If a non-adaptive algorithm $\mathcal{A}$ for Markov chain $M$ uses $Q$ queries and runs in time $Time$, then the induced non-adaptive algorithm $\widehat{\mathcal{A}}$ for $\Mh$ can be implemented with $\tilde O(kQ^2)$ queries and total running time $Time(\mathcal{A}) + \tilde O(kQ^2)$, where $k=\card{\Sigma}$. Also, any non-adaptive $\mAh$ for $\Mh$ with $Q'$ queries and time $Time(\mAh)$ yields a non-adaptive algorithm for $M$ with $\Omega(Q'/k)$ queries and running time $Time(\mAh)+ O(Q')$.
\end{remark}

\clearpage
\bibliographystyle{plainnat}

\bibliography{references}

@inproceedings{BehnezhadRR-STOC23,
  author       = {Soheil Behnezhad and
                  Mohammad Roghani and
                  Aviad Rubinstein},
  title        = {Sublinear Time Algorithms and Complexity of Approximate Maximum Matching},
  booktitle    = {Proceedings of the 55th Annual {ACM} Symposium on Theory of Computing,
                  {STOC} 2023, Orlando, FL, USA, June 20-23, 2023},
  pages        = {267--280},
  year         = {2023},
  url          = {https://doi.org/10.1145/3564246.3585231},
  doi          = {10.1145/3564246.3585231},
  bibsource    = {dblp computer science bibliography, https://dblp.org}
}

@inproceedings{BehnezhadRR-FOCS23,
  author       = {Soheil Behnezhad and
                  Mohammad Roghani and
                  Aviad Rubinstein},
  title        = {Local Computation Algorithms for Maximum Matching: New Lower Bounds},
  booktitle    = {64th {IEEE} Annual Symposium on Foundations of Computer Science, {FOCS}
                  2023, Santa Cruz, CA, USA, November 6-9, 2023},
  pages        = {2322--2335},
  year         = {2023},
  url          = {https://doi.org/10.1109/FOCS57990.2023.00143},
  doi          = {10.1109/FOCS57990.2023.00143},
  bibsource    = {dblp computer science bibliography, https://dblp.org}
}

@inproceedings{BehnezhadRR-STOC24,
  author       = {Soheil Behnezhad and
                  Mohammad Roghani and
                  Aviad Rubinstein},
  title        = {Approximating Maximum Matching Requires Almost Quadratic Time},
  booktitle    = {Proceedings of the 56th Annual {ACM} Symposium on Theory of Computing,
                  {STOC} 2024, Vancouver, BC, Canada, June 24-28, 2024},
  pages        = {444--454},
  year         = {2024},
  url          = {https://doi.org/10.1145/3618260.3649785},
  doi          = {10.1145/3618260.3649785},
  bibsource    = {dblp computer science bibliography, https://dblp.org}
}

@inproceedings{AldousVazirani,
  author       = {David J. Aldous and
                  Umesh V. Vazirani},
  title        = {"Go With the Winners" Algorithms},
  booktitle    = {35th Annual Symposium on Foundations of Computer Science, Santa Fe,
                  New Mexico, USA, 20-22 November 1994},
  pages        = {492--501},
  publisher    = {{IEEE} Computer Society},
  year         = {1994},
  url          = {https://doi.org/10.1109/SFCS.1994.365742},
  doi          = {10.1109/SFCS.1994.365742},
  bibsource    = {dblp computer science bibliography, https://dblp.org}
}

@inproceedings{CanettiGGM00,
  author       = {Ran Canetti and
                  Oded Goldreich and
                  Shafi Goldwasser and
                  Silvio Micali},
  editor       = {F. Frances Yao and
                  Eugene M. Luks},
  title        = {Resettable zero-knowledge (extended abstract)},
  booktitle    = {Proceedings of the Thirty-Second Annual {ACM} Symposium on Theory
                  of Computing, May 21-23, 2000, Portland, OR, {USA}},
  pages        = {235--244},
  publisher    = {{ACM}},
  year         = {2000},
  url          = {https://doi.org/10.1145/335305.335334},
  doi          = {10.1145/335305.335334},
  bibsource    = {dblp computer science bibliography, https://dblp.org}
}

@inproceedings{Schoning99,
  author       = {Uwe Sch{\"{o}}ning},
  title        = {A Probabilistic Algorithm for k-SAT and Constraint Satisfaction Problems},
  booktitle    = {40th Annual Symposium on Foundations of Computer Science, {FOCS} '99,
                  17-18 October, 1999, New York, NY, {USA}},
  pages        = {410--414},
  publisher    = {{IEEE} Computer Society},
  year         = {1999},
  url          = {https://doi.org/10.1109/SFFCS.1999.814612},
  doi          = {10.1109/SFFCS.1999.814612},
  bibsource    = {dblp computer science bibliography, https://dblp.org}
}

@article{GoldwasserMR89,
  author       = {Shafi Goldwasser and
                  Silvio Micali and
                  Charles Rackoff},
  title        = {The Knowledge Complexity of Interactive Proof Systems},
  journal      = {{SIAM} J. Comput.},
  volume       = {18},
  number       = {1},
  pages        = {186--208},
  year         = {1989},
  url          = {https://doi.org/10.1137/0218012},
  doi          = {10.1137/0218012},
  bibsource    = {dblp computer science bibliography, https://dblp.org}
}

@inproceedings{AssadiCK19,
  author       = {Sepehr Assadi and
                  Yu Chen and
                  Sanjeev Khanna},
  editor       = {Timothy M. Chan},
  title        = {Sublinear Algorithms for ({\(\Delta\)} + 1) Vertex Coloring},
  booktitle    = {Proceedings of the Thirtieth Annual {ACM-SIAM} Symposium on Discrete
                  Algorithms, {SODA} 2019, San Diego, California, USA, January 6-9,
                  2019},
  pages        = {767--786},
  publisher    = {{SIAM}},
  year         = {2019},
  url          = {https://doi.org/10.1137/1.9781611975482.48},
  doi          = {10.1137/1.9781611975482.48},
  bibsource    = {dblp computer science bibliography, https://dblp.org}
}

@inproceedings{Assadi022,
  author       = {Sepehr Assadi and
                  Chen Wang},
  editor       = {Mark Braverman},
  title        = {Sublinear Time and Space Algorithms for Correlation Clustering via
                  Sparse-Dense Decompositions},
  booktitle    = {13th Innovations in Theoretical Computer Science Conference, {ITCS}
                  2022, January 31 - February 3, 2022, Berkeley, CA, {USA}},
  series       = {LIPIcs},
  volume       = {215},
  pages        = {10:1--10:20},
  publisher    = {Schloss Dagstuhl - Leibniz-Zentrum f{\"{u}}r Informatik},
  year         = {2022},
  url          = {https://doi.org/10.4230/LIPIcs.ITCS.2022.10},
  doi          = {10.4230/LIPICS.ITCS.2022.10},
  bibsource    = {dblp computer science bibliography, https://dblp.org}
}

@inproceedings{CzumajS04,
  author       = {Artur Czumaj and
                  Christian Sohler},
  editor       = {L{\'{a}}szl{\'{o}} Babai},
  title        = {Estimating the weight of metric minimum spanning trees in sublinear-time},
  booktitle    = {Proceedings of the 36th Annual {ACM} Symposium on Theory of Computing,
                  Chicago, IL, USA, June 13-16, 2004},
  pages        = {175--183},
  publisher    = {{ACM}},
  year         = {2004},
  url          = {https://doi.org/10.1145/1007352.1007386},
  doi          = {10.1145/1007352.1007386},
  bibsource    = {dblp computer science bibliography, https://dblp.org}
}

@inproceedings{ChazelleRT01,
  author       = {Bernard Chazelle and
                  Ronitt Rubinfeld and
                  Luca Trevisan},
  editor       = {Fernando Orejas and
                  Paul G. Spirakis and
                  Jan van Leeuwen},
  title        = {Approximating the Minimum Spanning Tree Weight in Sublinear Time},
  booktitle    = {Automata, Languages and Programming, 28th International Colloquium,
                  {ICALP} 2001, Crete, Greece, July 8-12, 2001, Proceedings},
  series       = {Lecture Notes in Computer Science},
  volume       = {2076},
  pages        = {190--200},
  publisher    = {Springer},
  year         = {2001},
  url          = {https://doi.org/10.1007/3-540-48224-5\_16},
  doi          = {10.1007/3-540-48224-5\_16},
  bibsource    = {dblp computer science bibliography, https://dblp.org}
}

@inproceedings{CKT23,
  author       = {Yu Chen and
                  Sanjeev Khanna and
                  Zihan Tan},
  editor       = {Nikhil Bansal and
                  Viswanath Nagarajan},
  title        = {Query Complexity of the Metric Steiner Tree Problem},
  booktitle    = {Proceedings of the 2023 {ACM-SIAM} Symposium on Discrete Algorithms,
                  {SODA} 2023, Florence, Italy, January 22-25, 2023},
  pages        = {4893--4935},
  publisher    = {{SIAM}},
  year         = {2023},
  url          = {https://doi.org/10.1137/1.9781611977554.ch179},
  doi          = {10.1137/1.9781611977554.CH179},
  bibsource    = {dblp computer science bibliography, https://dblp.org}
}

@inproceedings{BehnezhadRRS-ICALP24,
  author       = {Soheil Behnezhad and
                  Mohammad Roghani and
                  Aviad Rubinstein and
                  Amin Saberi},
  editor       = {Karl Bringmann and
                  Martin Grohe and
                  Gabriele Puppis and
                  Ola Svensson},
  title        = {Sublinear Algorithms for {TSP} via Path Covers},
  booktitle    = {51st International Colloquium on Automata, Languages, and Programming,
                  {ICALP} 2024, July 8-12, 2024, Tallinn, Estonia},
  series       = {LIPIcs},
  volume       = {297},
  pages        = {19:1--19:16},
  publisher    = {Schloss Dagstuhl - Leibniz-Zentrum f{\"{u}}r Informatik},
  year         = {2024},
  url          = {https://doi.org/10.4230/LIPIcs.ICALP.2024.19},
  doi          = {10.4230/LIPICS.ICALP.2024.19},
  bibsource    = {dblp computer science bibliography, https://dblp.org}
}

@inproceedings{CKT-ICALP23,
  author       = {Yu Chen and
                  Sanjeev Khanna and
                  Zihan Tan},
  editor       = {Kousha Etessami and
                  Uriel Feige and
                  Gabriele Puppis},
  title        = {Sublinear Algorithms and Lower Bounds for Estimating {MST} and {TSP}
                  Cost in General Metrics},
  booktitle    = {50th International Colloquium on Automata, Languages, and Programming,
                  {ICALP} 2023, July 10-14, 2023, Paderborn, Germany},
  series       = {LIPIcs},
  volume       = {261},
  pages        = {37:1--37:16},
  publisher    = {Schloss Dagstuhl - Leibniz-Zentrum f{\"{u}}r Informatik},
  year         = {2023},
  url          = {https://doi.org/10.4230/LIPIcs.ICALP.2023.37},
  doi          = {10.4230/LIPICS.ICALP.2023.37},
  bibsource    = {dblp computer science bibliography, https://dblp.org}
}

@inproceedings{Behnezhad21,
  author       = {Soheil Behnezhad},
  title        = {Time-Optimal Sublinear Algorithms for Matching and Vertex Cover},
  booktitle    = {62nd {IEEE} Annual Symposium on Foundations of Computer Science, {FOCS}
                  2021, Denver, CO, USA, February 7-10, 2022},
  pages        = {873--884},
  publisher    = {{IEEE}},
  year         = {2021},
  url          = {https://doi.org/10.1109/FOCS52979.2021.00089},
  doi          = {10.1109/FOCS52979.2021.00089},
  bibsource    = {dblp computer science bibliography, https://dblp.org}
}

@inproceedings{YoshidaYI09,
  author       = {Yuichi Yoshida and
                  Masaki Yamamoto and
                  Hiro Ito},
  editor       = {Michael Mitzenmacher},
  title        = {An improved constant-time approximation algorithm for maximum matchings},
  booktitle    = {Proceedings of the 41st Annual {ACM} Symposium on Theory of Computing,
                  {STOC} 2009, Bethesda, MD, USA, May 31 - June 2, 2009},
  pages        = {225--234},
  publisher    = {{ACM}},
  year         = {2009},
  url          = {https://doi.org/10.1145/1536414.1536447},
  doi          = {10.1145/1536414.1536447},
  bibsource    = {dblp computer science bibliography, https://dblp.org}
}

@inproceedings{NguyenO08,
  author       = {Huy N. Nguyen and
                  Krzysztof Onak},
  title        = {Constant-Time Approximation Algorithms via Local Improvements},
  booktitle    = {49th Annual {IEEE} Symposium on Foundations of Computer Science, {FOCS}
                  2008, October 25-28, 2008, Philadelphia, PA, {USA}},
  pages        = {327--336},
  publisher    = {{IEEE} Computer Society},
  year         = {2008},
  url          = {https://doi.org/10.1109/FOCS.2008.81},
  doi          = {10.1109/FOCS.2008.81},
  bibsource    = {dblp computer science bibliography, https://dblp.org}
}

@inproceedings{BhattacharyaKS23,
  author       = {Sayan Bhattacharya and
                  Peter Kiss and
                  Thatchaphol Saranurak},
  title        = {Dynamic  (1+{$\varepsilon$})-Approximate Matching Size in Truly Sublinear
                  Update Time},
  booktitle    = {64th {IEEE} Annual Symposium on Foundations of Computer Science, {FOCS}
                  2023, Santa Cruz, CA, USA, November 6-9, 2023},
  pages        = {1563--1588},
  publisher    = {{IEEE}},
  year         = {2023},
  url          = {https://doi.org/10.1109/FOCS57990.2023.00095},
  doi          = {10.1109/FOCS57990.2023.00095},
  bibsource    = {dblp computer science bibliography, https://dblp.org}
}

@inproceedings{GoldreichR97,
  author       = {Oded Goldreich and
                  Dana Ron},
  editor       = {Frank Thomson Leighton and
                  Peter W. Shor},
  title        = {Property Testing in Bounded Degree Graphs},
  booktitle    = {Proceedings of the Twenty-Ninth Annual {ACM} Symposium on the Theory
                  of Computing, El Paso, Texas, USA, May 4-6, 1997},
  pages        = {406--415},
  publisher    = {{ACM}},
  year         = {1997},
  url          = {https://doi.org/10.1145/258533.258627},
  doi          = {10.1145/258533.258627},
  bibsource    = {dblp computer science bibliography, https://dblp.org}
}

@article{BenderR00,
  title={Testing properties of directed graphs: acyclicity and connectivity},
  author={Bender, Michael A and Ron, Dana},
  journal={Random Structures \& Algorithms},
  volume={20},
  number={2},
  pages={184--205},
  year={2002},
  publisher={Wiley Online Library}
}

@article{DumitriuTW03,
  author       = {Ioana Dumitriu and
                  Prasad Tetali and
                  Peter Winkler},
  title        = {On Playing Golf with Two Balls},
  journal      = {{SIAM} J. Discret. Math.},
  volume       = {16},
  number       = {4},
  pages        = {604--615},
  year         = {2003},
  url          = {https://doi.org/10.1137/S0895480102408341},
  doi          = {10.1137/S0895480102408341},
  bibsource    = {dblp computer science bibliography, https://dblp.org}
}

@article{JansonP12,
  author       = {Svante Janson and
                  Yuval Peres},
  title        = {Hitting Times for Random Walks with Restarts},
  journal      = {{SIAM} J. Discret. Math.},
  volume       = {26},
  number       = {2},
  pages        = {537--547},
  year         = {2012},
  url          = {https://doi.org/10.1137/100796352},
  doi          = {10.1137/100796352},
  bibsource    = {dblp computer science bibliography, https://dblp.org}
}

@inproceedings{MitrovicRS24,
  title={Locally Computing Edge Orientations},
  author={Mitrovi{\'c}, Slobodan and Rubinfeld, Ronitt and Singhal, Mihir},
  booktitle={32nd Annual European Symposium on Algorithms (ESA 2024)},
  pages={89--1},
  year={2024},
  organization={Schloss Dagstuhl--Leibniz-Zentrum f{\"u}r Informatik}
}

@inproceedings{Shah26,
  title={Sublinear-Time Lower Bounds for Approximating Matching Size using Non-Adaptive Queries},
  author={Shah, Vihan},
  booktitle={Proceedings of the 2026 Annual ACM-SIAM Symposium on Discrete Algorithms (SODA)},
  pages={5027--5065},
  year={2026},
  organization={SIAM}
}

@inproceedings{AzarmehrBGS25,
  author       = {Amir Azarmehr and
                  Soheil Behnezhad and
                  Alma Ghafari and
                  Madhu Sudan},
  title        = {Lower Bounds for Non-adaptive Local Computation Algorithms},
  booktitle    = {66th {IEEE} Annual Symposium on Foundations of Computer Science, {FOCS}
                  2025, Sydney, Australia, December 14-17, 2025},
  year         = {2025},
}

\appendix
\section{Further Connections to Sublinear-Time Graph Algorithms}
\label{sec:sublinear}

In this section, we further elaborate on the connections to sublinear-time graph algorithms by giving a concrete example where the lower bound is captured by Markov chains with rewinding.

We examine the lower bound of \citet*{BenderR00} for testing acyclicity in directed graphs.
Given adjacency list access, the algorithm must determine whether the input graph is acyclic or $\epsilon$-far from acyclic (i.e., at least an $\epsilon$-fraction of the edges must be removed so that no directed cycles remain in the graph) with a probability of $\frac{2}{3}$.
They prove that, for $\epsilon \leq \frac{1}{16}$, any such algorithm requires at least $\Omega(n^{1/3})$ queries to the adjacency list.
We show that this lower bound is captured by our model.
First, we examine the case where the algorithm has access only to the outgoing edges of a vertex, and later extend it to a setting where the algorithm has access to both the outgoing and incoming adjacency lists at the same time.

We start by reviewing the construction of \cite{BenderR00}.
By Yao's minimax lemma, to prove the lower bound, it suffices to construct two input distributions $\DYes$ and $\DNo$, respectively consisting of acyclic and far-from-acyclic graphs, such that any deterministic algorithm with $O(n^{1/3})$ queries fails to distinguish between them with a sufficiently large probability.
The construction uses a constant degree-parameter $d \geq 128$ shared between $\DYes$ and $\DNo$.

The $\DNo$ distribution constructs a graph by randomly partitioning the vertices into two groups, $V_U$ and $V_D$, each of size $n/2$. For the edges, a bipartite $d$-regular graph between $V_U$ and $V_D$ is chosen uniformly at random and directed towards $V_D$, and another one is chosen similarly and directed towards $V_U$. As a result, each vertex in the graph has both in-degree and out-degree equal to $d$, and the graph is $\epsilon$-far from acyclic with a high probability of $1 - 2^{-n}$, when $\epsilon < \frac{1}{16}$ and $d \geq 128$ \cite[Lemma 5]{BenderR00}.

The $\DYes$ distribution constructs a graph by randomly partitioning the vertices into $n^{1/3}$ groups, $L_1, L_2, \ldots, L_{n^{1/3}}$, each of size $n^{2/3}$.
For the edges, for each $1 \leq i < n^{1/3}$, a $d$-regular graph between $L_i$ and $L_{i+1}$ is chosen uniformly at random and directed towards $L_{i + 1}$.
As a result, the graph has no directed cycles and every vertex, except those of $L_1$ and $L_{n^{1/3}}$, has an in-degree and out-degree equal to $d$. The distributions are depicted in \cref{fig:acyclicity-graphs}.

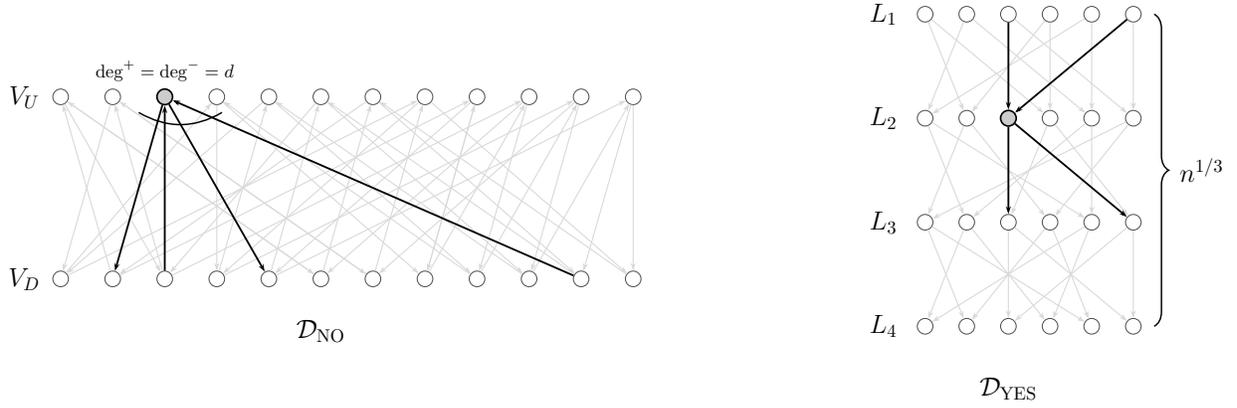
\begin{figure}
    \centering
    \resizebox{\textwidth}{!}{%
        \tikz[baseline]{
            \node {\begin{tikzpicture}[
    vertex/.style={circle, draw=black!70, fill=white, minimum size=3mm, inner sep=0pt, line width=0.5pt},
    edge/.style={-{Stealth[length=1.5mm, width=1mm]}, black!15, line width=0.6pt},
    highlight/.style={circle, draw=black, fill=black!20, minimum size=3mm, inner sep=0pt, line width=0.9pt},
    blackedge/.style={-{Stealth[length=1.5mm, width=1mm]}, black, line width=1pt}
]

\def\numVertices{12}
\def\hSpacing{1.0}  
\def\vSpacing{3.5}  
\def\highlightVertex{3}

\pgfmathsetmacro{\topRow}{\vSpacing}
\pgfmathsetmacro{\bottomRow}{0}

\foreach \i in {1,...,\numVertices} {
    \ifnum\i=\highlightVertex
        \node[highlight] (u\i) at (\i*\hSpacing, \topRow) {};
    \else
        \node[vertex] (u\i) at (\i*\hSpacing, \topRow) {};
    \fi
}

\foreach \i in {1,...,\numVertices} {
    \node[vertex] (d\i) at (\i*\hSpacing, \bottomRow) {};
}

\draw[edge] (u1) -- (d3);
\draw[edge] (u1) -- (d6);
\draw[edge] (u2) -- (d1);
\draw[edge] (u2) -- (d7);
\draw[edge] (u4) -- (d4);
\draw[edge] (u4) -- (d8);
\draw[edge] (u5) -- (d1);
\draw[edge] (u5) -- (d9);
\draw[edge] (u6) -- (d3);
\draw[edge] (u6) -- (d10);
\draw[edge] (u7) -- (d2);
\draw[edge] (u7) -- (d11);
\draw[edge] (u8) -- (d4);
\draw[edge] (u8) -- (d12);
\draw[edge] (u9) -- (d5);
\draw[edge] (u9) -- (d8);
\draw[edge] (u10) -- (d6);
\draw[edge] (u10) -- (d9);
\draw[edge] (u11) -- (d7);
\draw[edge] (u11) -- (d10);
\draw[edge] (u12) -- (d11);
\draw[edge] (u12) -- (d12);

\draw[edge] (d1) -- (u4);
\draw[edge] (d1) -- (u8);
\draw[edge] (d2) -- (u1);
\draw[edge] (d2) -- (u9);
\draw[edge] (d3) -- (u2);
\draw[edge] (d3) -- (u10);
\draw[edge] (d4) -- (u5);
\draw[edge] (d4) -- (u11);
\draw[edge] (d5) -- (u6);
\draw[edge] (d5) -- (u12);
\draw[edge] (d6) -- (u1);
\draw[edge] (d6) -- (u7);
\draw[edge] (d7) -- (u2);
\draw[edge] (d7) -- (u8);
\draw[edge] (d8) -- (u9);
\draw[edge] (d8) -- (u10);
\draw[edge] (d9) -- (u4);
\draw[edge] (d9) -- (u11);
\draw[edge] (d10) -- (u5);
\draw[edge] (d10) -- (u12);
\draw[edge] (d11) -- (u6);
\draw[edge] (d12) -- (u7);

\draw[blackedge] (u3) -- (d2);
\draw[blackedge] (u3) -- (d5);

\draw[blackedge] (d3) -- (u3);
\draw[blackedge] (d11) -- (u3);

\pgfmathsetmacro{\labelX}{0.3}
\node[font=\Large] at (\labelX, \topRow) {$V_U$};
\node[font=\Large] at (\labelX, \bottomRow) {$V_D$};

\pgfmathsetmacro{\arcLeft}{(\highlightVertex-0.5)*\hSpacing}
\pgfmathsetmacro{\arcRight}{(\highlightVertex+1.1)*\hSpacing}
\pgfmathsetmacro{\arcY}{\topRow-0.3}
\draw[thick, black] (\arcLeft, \arcY) to[out=-30, in=-150] (\arcRight, \arcY);

\pgfmathsetmacro{\textX}{\highlightVertex*\hSpacing}
\pgfmathsetmacro{\textY}{\topRow+0.5}
\node[font=\small, black] at (\textX, \textY) {$\deg^+ = \deg^- = d$};

\pgfmathsetmacro{\nameX}{(\numVertices*\hSpacing)/2}
\pgfmathsetmacro{\nameY}{\bottomRow-1.0}
\node[font=\Large] at (\nameX, \nameY) {$\DNo$};

\end{tikzpicture}};
        }\hspace{4cm}%
        \tikz[baseline]{
            \node {\begin{tikzpicture}[
    vertex/.style={circle, draw=black!70, fill=white, minimum size=3mm, inner sep=0pt, line width=0.5pt},
    edge/.style={-{Stealth[length=1.5mm, width=1mm]}, black!15, line width=0.6pt},
    highlight/.style={circle, draw=black, fill=black!20, minimum size=3mm, inner sep=0pt, line width=0.9pt},
    blackedge/.style={-{Stealth[length=1.5mm, width=1mm]}, black, line width=1pt}
]

\def\numLayers{4}
\def\numVertices{6}
\def\hSpacing{0.8}  
\def\vSpacing{2}    
\def\highlightLayer{2}
\def\highlightVertex{3}

\pgfmathsetmacro{\layerOne}{(\numLayers-1)*\vSpacing}
\pgfmathsetmacro{\layerTwo}{(\numLayers-2)*\vSpacing}
\pgfmathsetmacro{\layerThree}{(\numLayers-3)*\vSpacing}
\pgfmathsetmacro{\layerFour}{0}

\foreach \i in {1,...,\numVertices} {
    \node[vertex] (l1-\i) at (\i*\hSpacing, \layerOne) {};
}

\foreach \i in {1,...,\numVertices} {
    \ifnum\i=\highlightVertex
        \node[highlight] (l2-\i) at (\i*\hSpacing, \layerTwo) {};
    \else
        \node[vertex] (l2-\i) at (\i*\hSpacing, \layerTwo) {};
    \fi
}

\foreach \i in {1,...,\numVertices} {
    \node[vertex] (l3-\i) at (\i*\hSpacing, \layerThree) {};
}

\foreach \i in {1,...,\numVertices} {
    \node[vertex] (l4-\i) at (\i*\hSpacing, \layerFour) {};
}

\draw[edge] (l1-1) -- (l2-2);
\draw[edge] (l1-1) -- (l2-4);
\draw[edge] (l1-2) -- (l2-1);
\draw[edge] (l1-2) -- (l2-5);
\draw[edge] (l1-3) -- (l2-6);
\draw[edge] (l1-4) -- (l2-2);
\draw[edge] (l1-4) -- (l2-4);
\draw[edge] (l1-5) -- (l2-1);
\draw[edge] (l1-5) -- (l2-5);
\draw[edge] (l1-6) -- (l2-6);

\draw[blackedge] (l1-3) -- (l2-3);
\draw[blackedge] (l1-6) -- (l2-3);

\draw[edge] (l2-1) -- (l3-2);
\draw[edge] (l2-1) -- (l3-5);
\draw[edge] (l2-2) -- (l3-1);
\draw[edge] (l2-2) -- (l3-4);
\draw[edge] (l2-4) -- (l3-3);
\draw[edge] (l2-4) -- (l3-6);
\draw[edge] (l2-5) -- (l3-2);
\draw[edge] (l2-5) -- (l3-5);
\draw[edge] (l2-6) -- (l3-1);
\draw[edge] (l2-6) -- (l3-4);

\draw[blackedge] (l2-3) -- (l3-3);
\draw[blackedge] (l2-3) -- (l3-6);

\draw[edge] (l3-1) -- (l4-2);
\draw[edge] (l3-1) -- (l4-5);
\draw[edge] (l3-2) -- (l4-1);
\draw[edge] (l3-2) -- (l4-4);
\draw[edge] (l3-3) -- (l4-3);
\draw[edge] (l3-3) -- (l4-6);
\draw[edge] (l3-4) -- (l4-2);
\draw[edge] (l3-4) -- (l4-5);
\draw[edge] (l3-5) -- (l4-1);
\draw[edge] (l3-5) -- (l4-4);
\draw[edge] (l3-6) -- (l4-3);
\draw[edge] (l3-6) -- (l4-6);

\node[font=\Large] at (0, \layerOne) {$L_1$};
\node[font=\Large] at (0, \layerTwo) {$L_2$};
\node[font=\Large] at (0, \layerThree) {$L_3$};
\node[font=\Large] at (0, \layerFour) {$L_4$};

\pgfmathsetmacro{\bracketX}{(\numVertices+0.5)*\hSpacing}
\draw[thick, black, decorate, decoration={brace, amplitude=8pt, mirror}] 
    (\bracketX, \layerFour) -- (\bracketX, \layerOne) 
    node[midway, right=10pt, font=\Large] {$n^{1/3}$};

\pgfmathsetmacro{\nameX}{(\numVertices*\hSpacing)/2}
\pgfmathsetmacro{\nameY}{-1.2}
\node[font=\Large] at (\nameX, \nameY) {$\DYes$};

\end{tikzpicture}};
        }%
    }
    \caption{Sample graphs from the $\DYes$ (acyclic) and $\DNo$ (far from acyclic) distributions. The typical vertex has an in-degree and out-degree of $d$.}
    \label{fig:acyclicity-graphs}
\end{figure}

Now, we present the Markov chain that captures this lower-bound construction.
The Markov chain consists of two parts: one for $\DNo$ and one for $\DYes$.
In the $\No$ part, each of the vertex groups $V_U$ and $V_D$ is represented by a state, $v_U$ and $v_D$ respectively. Considering that in the graph, the outgoing edges of $V_U$ go to $V_D$ and the outgoing edges of $V_D$ go to $V_U$, in the Markov chain, we let $v_D$ transition to $v_U$, and let $v_U$ transition to $v_D$ with probability $1$.
To model starting at a random vertex, we use an initial state $x_{\No}$ that transitions to $v_D$ or $v_U$ with probability $\frac{1}{2}$ each.
Similarly, in the $\Yes$ part, each vertex group $L_i$ is represented by a state $\ell_i$. Each state $\ell_i$ for $1 \leq i < n^{1/3}$, transitions to the $\ell_{i+1}$ with probability $1$.
To model starting at a random vertex, we use an initial state $x_{\Yes}$ that transitions to each of $\ell_1, \ell_2 \ldots, \ell_{n^{1/3}}$ with probability $\frac{1}{n^{1/3}}$.
As the algorithm interacts with the Markov chain, the only observation it makes is whether the drawn state is $\ell_{n^{1/3}}$ or not. The Markov chain is illustrated in \cref{fig:acyclicity-chain}. 

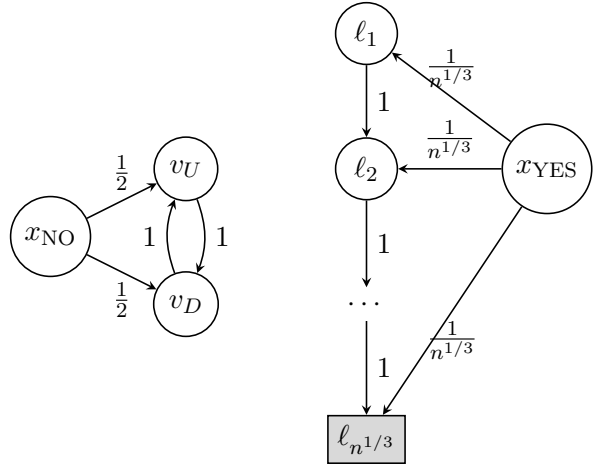
\begin{figure}
    \centering
    \begin{tikzpicture}[->, >=stealth, auto, semithick, scale=0.6]
  \def\statesize{0cm}   
  
  \tikzstyle{every state}=[fill=none,draw=black,text=black,minimum size=\statesize]
  \tikzstyle{initial state}=[circle,fill=none,draw=black,text=black,minimum size=\statesize]
  \tikzstyle{final state}=[fill=gray!30,draw=black,text=black,minimum size=\statesize]
  \tikzstyle{dots state}=[draw=none,text=black,minimum size=\statesize]
  
  \def\nohspace{3}      
  \def\novspace{3}      
  \def\yeshspace{4}     
  \def\yesvspace{3}     
  \def\casesep{4}       
  
  \node[initial state] (xno) at (-\nohspace, 0.5*\novspace) {$x_{\No}$};
  \node[state] (vu) at (0, \novspace) {$v_U$};
  \node[state] (vd) at (0, 0) {$v_D$};
  
  \path (xno) edge node[above] {$\frac{1}{2}$} (vu)
        (xno) edge node[below] {$\frac{1}{2}$} (vd)
        (vu) edge[bend left=20] node[right] {$1$} (vd)
        (vd) edge[bend left=20] node[left] {$1$} (vu);
  
  \node[initial state] (xyes) at (\casesep+\yeshspace, \yesvspace) {$x_{\Yes}$};
  \node[state] (l1) at (\casesep, 2*\yesvspace) {$\ell_1$};
  \node[state] (l2) at (\casesep, \yesvspace) {$\ell_2$};
  \node[dots state] (ldots) at (\casesep, 0) {$\cdots$};
  \node[final state] (ln) at (\casesep, -\yesvspace) {$\ell_{n^{1/3}}$};
  
  \path (xyes) edge node[above] {$\frac{1}{n^{1/3}}$} (l1)
        (xyes) edge node[above] {$\frac{1}{n^{1/3}}$} (l2)
        (xyes) edge node[below] {$\frac{1}{n^{1/3}}$} (ln)
        (l1) edge node[right] {$1$} (l2)
        (l2) edge node[right] {$1$} (ldots)
        (ldots) edge node[right] {$1$} (ln);
\end{tikzpicture}
    \caption{A Markov chain that models exploring the input graph drawn from the $\DYes$ and $\DNo$ distributions. Distinguishing the initial states $x_\Yes$ and $x_\No$ is equivalent to distinguishing between the two input distributions.}
    \label{fig:acyclicity-chain}
\end{figure}

Next, we elaborate on how this Markov chain captures the lower bound.
The key idea is that an algorithm that makes $O(n^{1/3})$ queries, does not discover any vertex more than once, with a sufficiently high constant probability. That is, the set of edges discovered by the sublinear algorithm does not contain any cycles, even when the directions are ignored. Let $(u_1, i_1), \ldots, (u_q, i_q)$ be the set of queries made to the adjacency lists. We assume without loss of generality that the queries are unique (i.e., the same query is not made twice), and that the degree of each vertex is revealed to the algorithm as it is discovered.
\begin{lemma}[{\cite[Lemma 7]{BenderR00}}]
    Consider any algorithm that makes $q \leq \frac{1}{4} \cdot n^{1/3}$ adjacency-list queries, and let $v_1, \ldots v_q$ be the answers, where $v_k \in V \cup \{\bot\}$. Then, with probability $1 - \frac{1}{16}$, no vertex appears in the answers more than once.
\end{lemma}

As a result, we can assume without loss of generality that no vertex appears twice in the answers, since $\frac{1}{16}$ is a negligible probability and when a vertex appears twice, we can simply presume that the algorithm successfully distinguishes between the input distributions.
Making this assumption, we argue that distinguishing between the input distributions $\DYes$ and $\DNo$ is equivalent to distinguishing between initial states $x_\Yes$ and $x_\No$ in the Markov chain.
Since there are no repeated vertices, the subgraph explored by the algorithm is a set of rooted trees.
When the algorithm explores the neighbors of a vertex $u$ in a group of vertices $A$, the discovered neighbor $v$ is a random vertex from a neighboring group $B$.
The distribution of $B$ is solely determined by $A$ and does not depend on $u$.\footnote{in this construction, $B$ is in fact a deterministic function of $A$.} 
Furthermore, the algorithm observes only the out-degree of each discovered vertex (which is different only for vertices in $L_{n^{1/3}}$), and not its group.
This is paralleled by the transition probabilities in the Markov chain, and the partial observation of the state, which only differentiates $\ell_{n^{1/3}}$.

As a result, to prove the lower bound for acyclicity testing, it suffices to prove the following in the context of Markov chains:
\begin{lemma}
    Any algorithm that distinguishes between the initial states $x_\Yes$ and $x_\No$ with a probability of $5/8$, must make $\frac{1}{4} \cdot n^{1/3}$ queries to the Markov chain.
\end{lemma}

The proof in \cite{BenderR00} essentially argues the same thing, without ever explicitly defining a Markov chain. We give a high-level overview of their proof here.
Since the only observation the algorithm can make is whether it has reached $\ell_{n^{1/3}}$ or not, it suffices to bound the probability of reaching $\ell_{n^{1/3}}$ when the initial state is $x_\Yes$.
The next state drawn from $x_{\Yes}$ is always a state $\ell_i$, where $i \in \{1, 2, \ldots, n^{1/3}\}$ is uniformly random.
As a result, exploring a tree of depth $D$ from such a state, reaches $\ell_{n^{1/3}}$ with probability of at most $\frac{D + 1}{n^{1/3}}$.
Therefore, since the sum of the depths of the explored trees is $O(n^{1/3})$, the probability of reaching $\ell_{n^{1/3}}$ throughout the algorithm can be bounded by a constant.

Finally, we show that the lower bound is captured by our model in the more general case, where the algorithm is allowed to query both the incoming and outgoing adjacency lists.
To do so, we use a simple gadget.
On the $\Yes$ side of the Markov chain, for each state $\ell_i$, we create two paths of length two:
A path $\ell_i \to d_i \to \ell_{i+1}$, corresponding to traversing the outgoing edges of $L_i$,
and a path $\ell_i \to u_i \to \ell_{i - 1}$, corresponding to traversing the incoming edges of $L_i$ (the $\No$ side is handled similarly).
Additionally, we extend the observations made by the algorithm.
Other than observing whether the current state is $\ell_{n^{1/3}}$, the algorithm can also observe 
whether the current state is 
(1) $\ell_{1}$, the first layer corresponding to vertices with zero in-degree,
(2) an auxiliary state $d_i$, indicating the traversal of an outgoing edge, and
(3) an auxiliary state $u_i$, indicating the traversal of an incoming edge.
This allows the algorithm to choose which of the directions it takes, while increasing the number of queries to the Markov chain only by a constant factor.
See \cref{sec:generality} for a discussion of Markov chains with a constant-sized observation alphabet.

\end{document}